\pdfoutput=1
\documentclass[english]{lipics-v2021}

\usepackage{amsmath}
\usepackage{mathtools}
\usepackage{amsfonts}
\usepackage{amsthm}
\usepackage{tikz}
\usepackage{environ}
\makeatletter
\newsavebox{\measure@tikzpicture}
\NewEnviron{scaletikzpicturetowidth}[1]{%
  \def\tikz@width{#1}%
  \def\tikzscale{1}\begin{lrbox}{\measure@tikzpicture}%
  \BODY
  \end{lrbox}%
  \pgfmathparse{#1/\wd\measure@tikzpicture}%
  \edef\tikzscale{\pgfmathresult}%
  \BODY
}
\makeatother

\usepackage{algorithm}
\usepackage{algorithmicx}
\usepackage{algpseudocode}

\algtext*{EndWhile}
\algtext*{EndFor}
\algtext*{EndIf}
\algtext*{EndFunction}

\title{Embedding Graphs as Euclidean \texorpdfstring{$k$}{k}NN-Graphs}

\Copyright{Thomas Schibler, Subhash Suri, and Jie Xue}
\author{Thomas Schibler}{University of California, Santa Barbara, CA, USA}{tschibler@ucsb.edu}{https://orcid.org/0009-0008-2966-9468}{}
\author{Subhash Suri}{University of California Santa Barbara, CA, USA}{suri@ucsb.edu}{https://orcid.org/0000-0002-5668-7521}{}
\author{Jie Xue}{New York University Shanghai, China}{jiexue@nyu.edu}{https://orcid.org/0000-0001-7015-1988}{}
\authorrunning{T. Schibler, S. Suri, and J. Xue}

\hideLIPIcs
\nolinenumbers
\begin{document}

\maketitle

\begin{abstract}
Let $G=(V,E)$ be a directed graph on $n$ vertices where each vertex has out-degree $k$.
We say that $G$ is $k$NN-realizable in $d$-dimensional Euclidean space if there exists a point set $P = \{p_1, p_2, \ldots, p_n\}$ in $\mathbb{R}^d$ along with a one-to-one mapping $\phi: V \rightarrow P$ such that for any $u,v \in V$, $u$ is an out-neighbor of $v$ in $G$ if and only if $\phi(u)$ is one of the $k$ nearest neighbors of $\phi(v)$; we call the map $\phi$ a \textit{$k$NN-realization} of $G$ in $\mathbb{R}^d$.
The $k$NN-realization problem, which aims to compute a $k$NN-realization of an input graph in $\mathbb{R}^d$, is known to be NP-hard already for $d=2$ and $k=1$ [Eades and Whitesides, Theoretical Computer Science, 1996], and to the best of our knowledge has not been studied in dimension $d=1$. The main results of this paper are the following:
\smallskip
\begin{itemize}
    \item For any fixed dimension $d \geq 2$, we can efficiently compute an embedding realizing at least a $1 - \varepsilon$ fraction of $G$'s edges, or conclude that $G$ is not $k$NN-realizable in $\mathbb{R}^d$.
    \smallskip
    \item For $d=1$, we can decide in $O(kn)$ time whether $G$ is $k$NN-realizable and, if so, compute a realization in $O(n^{2.5} \mathsf{poly}(\log n))$ time.
\end{itemize}
\end{abstract}

\section{Introduction}

The \textit{$k$NN-graph} of a set $P$ of points in $\mathbb{R}^d$ is a directed graph with vertex set $P$ and edges defined as follows: there is a directed edge from $a \in P$ to $b \in P$ if and only if $b$ is one of the $k$-nearest neighbors of $a$ in $P$ (under the Euclidean distance).
A directed graph $G = (V, E)$ is \textit{$k$NN-realizable} in $\mathbb{R}^d$ if it is isomorphic to the $k$NN-graph of a set of points in $\mathbb{R}^d$.
In this paper, we consider the following natural problems regarding $k$NN-graphs:
Can we efficiently check whether a given directed graph is $k$NN-realizable in $\mathbb{R}^d$?
If so, can we efficiently find a $k$NN-realization of $G$?
More formally, a $k$NN-\emph{realization} (or $k$NN-\emph{embedding}) of $G$ in $\mathbb{R}^d$ is a one-to-one mapping $\phi: V \rightarrow P$ from the vertex set $V$ of $G$ to a set $P$ of points in $\mathbb{R}^d$ that induces an isomorphism between $G$ and the $k$NN-graph of $P$.
Throughout the paper, we use the terms embedding and realization interchangeably.

The $k$NN-graph is a member of the well-known family of \emph{proximity} graphs in computational geometry that includes minimum spanning trees, relative neighborhood graphs, Gabriel graphs, and Delaunay triangulations~\cite{cgbook}. Over the past several decades, a substantial research effort has been directed to design efficient algorithms to compute these structures for an input set of points.  The \emph{inverse} problem -- the focus of our paper -- where we want to recover the points that produce a given proximity graph, however, remains  less well understood.

An early example of a positive result in this direction is on Euclidean minimum spanning trees. Monma and Suri~\cite{monma:1992} show that any tree with maximum node degree $5$ can be realized as a minimum spanning tree of points in the plane and, additionally, every planar point set admits a minimum spanning tree with degree at most $5$. This settles the above problem for \textit{non-degenerate} planar point sets but if we allow co-circularities, then a planar minimum spanning tree can have maximum node degree $6$; in that case, the problem of computing a 1NN-realization was proved to be NP-complete by Eades and Whitesides~\cite{eades:1994}.
Dillencourt~\cite{dillencourt:1990} considers the problem of recognizing triangulations that can be realized as Delaunay triangulations in the plane, and it is also known that all outerplanar triangulations are realizable~\cite{sugihara:1994}. (These proofs are non-constructive, however, and only exponential time algorithms are known for computing the coordinates of the points in the realization~\cite{agrawal:2022}.)
Among other related results, Bose et al.~\cite{Jit-prox} give a characterization of  \emph{trees} that can be realized in the plane as relative neighborhood or Gabriel graphs, and Eppstein et al.~\cite{Epps-nn} establish some  structural properties of $k$NN graphs of random points in the plane.

The $k$NN-realization is not directly related to the
\emph{metric embedding} problem but there are some obvious similarities. The input to metric embedding is a weighted graph satisfying triangle inequalities and the goal is to find an embedding realizing all pairwise distances. This is not always possible~\cite{matousek-book} -- there are simple metrics that cannot be embedded in any finite dimensional Euclidean space. However, if we allow some \emph{distortion} (i.e. approximation) of distances, then any $n$-node metric graph can be embedded in $O(\log n)$ dimensional space with polylog distortion~\cite{bourgain}, or using the celebrated Johnson-Lindenstrauss theorem~\cite{johnson:1984} in dimension $O(\log n /\varepsilon^2 )$ with distortion $1 + \varepsilon$ for any fixed constant $\varepsilon > 0$.

In metric embedding the goal is to realize all pairwise distances (approximately), while in $k$NN-realization the goal is only to preserve all \emph{ordinal} neighbor relations in a \emph{directed} graph: a neighbor of each vertex must be closer than any of its non-neighbors but the choice of specific distances does not matter.
Although  there is some work in embedding ordinal relations as well, the main focus is specialized metric spaces and \emph{bounds} on \emph{ordinal distortion}.
In particular, Alon et al. \cite{alon2008} consider embeddings with ordinal relaxations into ultrametric spaces, and Bădoiu et al. \cite{badoiu2008} improve their bounds for tree and line embeddings.
A different line of research concerns \emph{triplets} embedding, where the input is a set of triples $(a, b, c)$ specifying constraints of the form $d(a, b) < d(a, c)$.
Even for embedding in one dimension $\mathbb{R}^1$, the general triplet problem as well as the related ``betweenness'' problem is MAXSNP-hard~\cite{ChorSudan}, and a FPTAS is known for maximizing the number of satisfied triplet constraints~\cite{fan2020}.
(If all $\binom{n}{3}$ triplet constraints are given, then the problem is trivial to solve -- indeed, once the leftmost point is determined, the rest of the ordering can be easily decided.)
In contrast with the general triplet constraints, the pairwise relations in a $k$NN-realization problem in $\mathbb{R}^1$ have a richer structure that allows us to solve the problem efficiently. 

The $k$NN-realization is also related to the \emph{sphericity} of graphs~\cite{maehara:1984}, where the goal is to embed a (undirected) graph $G=(V,E)$ into an Euclidean space $\mathbb{R}^k$ such that there is a point $p(u)$ for each vertex $u \in V$, and $d(p(u), p(v)) \leq 1$ iff $(u, v) \in E$.  The smallest dimension $k$ that admits such an embedding is called the \emph{sphericity} of $G$~\cite{maehara:1984}. (Another related problem is the \emph{dimension} of a graph, introduced by Erd\H{o}s, Harary and Tutte~\cite{erdos:65}: it is the smallest number $k$ such that $G$ can be embedded into Euclidean $k$-space with every edge of $G$ having length $1$.) Along these lines, Chatziafratis and Indyk~\cite{chatziafratis2024} also show that if we want to preserve the relative distances among the $k$ nearest neighbors of each point, then the embedding dimension must grow linearly with $k$. Unlike these results, our work is algorithmic -- we design efficient algorithms to realize a given graph $G$ in a specified dimension $d$.

Euclidean embedding of neighbor relations is also studied in social sciences for geometric realization of preference orderings. Given a set $V$ of $n$ voters, a set $C$ of $m$ candidates, and a rank ordering of the candidates by each voter, we say that the preference graph can be realized in Euclidean $d$-space if each voter's preferences are consistent with its Euclidean distances to all the candidates. In this case, it is known that the smallest dimension must satisfy $d \geq \min \{n , m-1\}$~\cite{bogomolnaia:2007}.

\subsection{Main results}

We study the \textit{$k$NN-realization} problem, which takes a directed graph $G$ as an input and aims to compute a $k$NN-realization of $G$ in $\mathbb{R}^d$ (or decide the nonexistence of such a realization).
Throughout we focus on the \emph{unranked} $k$NN-realization problem where nearest neighbors are realized as an \emph{unordered set} but our results also hold (with minor caveat) for the ranked version where edges of $G$ specify each of the $k$ neighbors in order (see Concluding Remarks). 
Our two main results are the following.

\begin{enumerate}
\item Assuming that $G$ is $k$NN-realizable in $\mathbb{R}^d$, we can find a $d$-dimensional realization in polynomial time that preserves at least a $1 - \varepsilon$ fraction of the edges of $G$, for any fixed $\varepsilon > 0$.
If our algorithm fails, we can also conclude that $G$ is \emph{not} $k$NN-realizable in $\mathbb{R}^d$.
In particular, our algorithm is an EPTAS (\textit{efficient polynomial-time approximation scheme}) for the $k$NN-realization problem for fixed $k$ and $d$.

\smallskip

\item We give a \textit{linear-time} algorithm to decide wether $G$ is $k$NN-realizable in $\mathbb{R}^1$.
Specifically, the algorithm runs in $O(kn)$ time, which is linear in the size of the input graph.
If $G$ is realizable, our algorithm can compute a realization in $O(n^{2.5} \mathsf{poly}(\log n))$ time.
\end{enumerate}

\subsection{Basic definitions}

For a (directed or undirected) graph $G$, we use $V(G)$ and $E(G)$ to denote the set of its vertices and edges, respectively. If $G$ is a directed graph, we say $G$ is \textit{$k$-regular} if the out-degree of every vertex of $G$ is exactly equal to $k$. We use $\mathsf{in}[v]$ (resp., $\mathsf{out}[v]$) to denote the set consisting of $v$ itself and all in-neighbors (resp., out-neighbors) of $v$.
A \textit{$k$NN-realization} of $G = (V,E)$ in a metric space $\mathcal{M} = (M,\mathsf{dist})$ is an injective map $\phi: V \rightarrow M$ such that  $\mathsf{dist}(\phi(v),\phi(u)) < \mathsf{dist}(\phi(v),\phi(u'))$, for any distinct triple $v,u,u' \in V$ where $(v,u) \in E$ and $(v,u') \notin E$.
We say $G$ is \textit{$k$NN-realizable} in $\mathcal{M}$ if there exists a $k$NN-realization of $G$ in $\mathcal{M}$.

\section{Approximate \texorpdfstring{$k$}{k}NN-realization in \texorpdfstring{$\mathbb{R}^d$}{Rd}}

We first observe that it is easy to decide in polynomial time whether a given graph $G$ is $k$NN-realizable in a finite-dimensional Euclidean space.
We just have to check the acyclicity of an auxiliary graph defined as follows. Let $\varLambda_G$ be a directed graph whose vertices correspond to pairs of vertices in $V(G)$ with a directed edge $(\{v,u\},\{v,u'\})$ for every distinct triple $v,u,u' \in V(G)$ where $(v,u) \in E(G)$ and $(v,u') \notin E(G)$.
Intuitively, $\varLambda_G$ encodes the $\leq$-relation among the pairwise distances of points in any (potential) $k$NN-realization of $G$.
If $\varLambda_G$ has a directed cycle, then clearly $G$ is not $k$NN-realizable.
Otherwise, a topological sort of $\varLambda_G$ gives a total ordering of the vertex pairs in $G$.
With this ordering, we can find a $k$NN-realization of $G$ in $\mathbb{R}^n$ using the result of Bilu and Linial~\cite{bilu2005}.

On the other hand, deciding whether $G$ is $k$NN-realizable in $\mathbb{R}^d$, for a specific dimension $d$, is $NP$-hard. This is shown by Eades and Whitesides~\cite{eades:1996} who proved the hardness for $d=2$ and $k=1$.
Therefore, in this section, we explore an \textit{approximation} algorithm for the realization problem in any fixed dimension $d$.
We first need to define an approximate solution to our problem.
There is a natural way to measure how well a map $\phi: V(G) \rightarrow \mathbb{R}^d$ approximates a $k$NN-realization: randomly sample an edge $(u,v) \in E(G)$ and consider the probability that $\phi(v)$ is among the $k$ nearest neighbors of $\phi(u)$.
Formally, we introduce the following definition.
\begin{definition}[approximate $k$NN-realization] \label{def-apprxreal}
Let $G$ be a $k$-regular directed graph.
For $c \in [0,1]$, a map $\phi:V(G) \rightarrow \mathbb{R}^d$ is a \textbf{$c$-approximate $k$NN-realization} of $G$ in $\mathbb{R}^d$ if 
\begin{equation*}
    \sum_{(u,v) \in E(G)} \sigma_\phi(u,v) \geq c \cdot |E(G)|,
\end{equation*}
where $\sigma_\phi: V(G) \times V(G) \rightarrow \{0,1\}$ is the indicator function defined as $\sigma_\phi(u,v) = 1$ if we have $|\{v' \in V(G) \backslash \{u\}: \lVert \phi(u)-\phi(v') \rVert_2 \leq \lVert \phi(u)-\phi(v) \rVert_2 \}| \leq k$ and $\sigma_\phi(u,v) = 0$ otherwise.
\end{definition}

Our main result is an algorithm for computing a $(1-\varepsilon)$-approximate $k$NN-realization of $G$ in $\mathbb{R}^d$ with time complexity $f(k,d,\varepsilon) \cdot n^{O(1)}$, for any given $\varepsilon > 0$ (provided that $G$  is $k$NN-realizable in $\mathbb{R}^d$), where $f$ is some computable function.
In other words, for fixed $k$ and $d$, we obtain an \textit{efficient polynomial-time approximation scheme} (EPTAS) for the $k$NN-realization problem in $\mathbb{R}^d$.

At a high level, our algorithm consists of two main steps.
In the first step, it computes a set $E \subseteq E(G)$ of edges such that $|E| \leq \varepsilon |E(G)|$ and each weakly-connected component of $G - E$ contains $O_\varepsilon(1)$ vertices.
The existence of $E$ follows from the result of Miller et al.~\cite{miller:1997} on graph separators, and the computation of $E$ relies on the approximate minimum cut algorithm of Chuzhoy et al.~\cite{chuzhoy2020deterministic}.
The edges in $E$ are the ones we sacrifice in our approximation and so, in the second step, our algorithm computes a map $\phi:V(G) \rightarrow \mathbb{R}^d$ satisfying $\sigma_\phi(u,v) = 1$ for all edges $(u,v) \in E(G) \backslash E$.
This turns out to be easy, since the weakly-connected components of $G - E$ are of size $O_\varepsilon(1)$ and we can work on these components individually.
These two steps will be presented in Sections~\ref{sec-split} and~\ref{sec-component}, respectively.

\subsection{Splitting the graph by removing few edges} \label{sec-split}

A \textit{balanced cut} of an \textit{undirected} graph $H$ is a subset $E \subseteq E(H)$ such that every connected component of $H- E$ contains at most $|V(H)|/2$ vertices.
Every directed graph $G$ naturally corresponds to an undirected graph $G_0$ defined as $V(G_0) = V(G)$ and $E(G_0) = \{\{u,v\}: (u,v) \in E(G) \text{ or } (v,u) \in E(G)\}$.
We need the following important result.
\begin{lemma} \label{lem-smallcut}
Let $G$ be a directed graph that is $k$NN-realizable in $\mathbb{R}^d$, and $G_0$ be the undirected graph corresponding to $G$.
Then any subgraph $H$ of $G_0$ with $|V(H)| \geq 2$ admits a balanced cut of size $O(|V(H)|^{1-\frac{1}{d}})$, where the constant hidden in $O(\cdot)$ only depends on $k$ and $d$.
\end{lemma}
\begin{proof}
Let $k,d \in \mathbb{N}$ be fixed numbers.
The lemma follows from two results in~\cite{miller:1997}.
Specifically, it was shown in \cite{miller:1997} that if a graph $G$ is $k$NN-realizable in $\mathbb{R}^d$, then its corresponding undirected graph $G_0$ satisfies the following conditions.
\begin{enumerate}
    \item The degree of every vertex in $G_0$ is $O(1)$.
    \item For every subgraph $H$ of $G_0$, there exists $S \subseteq V(H)$ such that $|S| = O(|V(H)|^{1-\frac{1}{d}})$ and every connected component of $H - S$ contains at most $\frac{d+1}{d+2} \cdot |V(H)|$ vertices.
\end{enumerate}
We remark that \cite{miller:1997} only claimed condition~2 above for the case $H = G_0$, but the argument extends to any subgraph $H$ of $G_0$.
Using the above two conditions, it is fairly easy to construct the desired balanced cut for any subgraph of $G_0$.
In~\cite{miller:1997}, condition~1 is shown in Corollary 3.2.3 and condition~2 is shown in Theorem 3.2.2.

We can construct a balanced cut $E \subseteq E(H)$ of a subgraph $H$ of $G_0$ as follows.
We begin with an empty $E$ and repeatedly apply the following operation until $E$ is a balanced cut of $H$.
Pick the \textit{largest} connected component $C$ of $H-E$, i.e., the component consisting of the maximum number of vertices.
Note that any connected component of $H-E$ other than $C$ contains at most $|V(H)|/2$ vertices.
If $|V(C)| \leq |V(H)|/2$, then $E$ is already a balanced cut.
Otherwise, we apply condition~2 above to find a set $S \subseteq V(C)$, and then add to $E$ all edges of $H$ incident to $S$.
The out-degree of each vertex in $H$ is at most $k = O(1)$, and the in-degree is also $O(1)$.
Thus, the number of edges added to $E$ is $O(|V(H)|^{1-\frac{1}{d}})$, since $|S| = O(|V(H)|^{1-\frac{1}{d}})$.
In this way, we iteratively add edges to $E$ until it becomes a balanced cut of $H$.
Observe that this procedure terminates in $O(1)$ iterations.
Indeed, by an induction argument and condition~2 above, one can easily show that after the $i$-th iteration, the largest connected component of $H-E$ contains at most $\max\{(\frac{d+1}{d+2})^i,\frac{1}{2}\} |V(H)|$ vertices.
Thus, after $O(1)$ iterations, $E$ is a balanced cut of $H$ and we have $|E| = O(|V(H)|^{1-\frac{1}{d}})$.
\end{proof}

A \textit{minimum} balanced cut refers to a balanced cut consisting of the minimum number of edges.
A \textit{$c$-approximate} minimum balanced cut refers to a balanced cut whose size is at most $c \cdot \mathsf{opt}$, where $\mathsf{opt}$ is the size of a minimum balanced cut.
The above lemma implies that if $G$ is $k$NN-realizable in $\mathbb{R}^d$, then the size of a minimum balanced cut of any subgraph $H$ of $G_0$ is $O(|V(H)|^{1-\frac{1}{d}})$.
We need the following algorithm by Chuzhoy et al. for computing approximate minimum balanced cuts.
\begin{theorem}[\cite{chuzhoy2020deterministic}] \label{thm-cut}
For any fixed number $\alpha > 0$, there exists an algorithm that, given an undirected graph $H$ of $n$ vertices and $m$ edges, computes a $\log^{O(1/\alpha)} n$-approximate minimum balanced cut of $H$ in $O(m^{1+\alpha})$ time.
\end{theorem}

The following lemma is the main result of this section, which states that one can remove a small fraction of edges from $G$ to split $G$ into small components.

\begin{lemma} \label{lem-edgeremoval}
Let $k, d \in \mathbb{N}$ be fixed numbers.
Given a $k$-regular directed graph $G$ and a number $\varepsilon \in (0,1]$, one can either compute a subset $E \subseteq E(G)$ such that $|E| \leq \varepsilon |E(G)|$ and each weakly-connected component of $G - E$ contains at most $(\frac{1}{\varepsilon})^{O(1)}$ vertices, or conclude that $G$ is not $k$NN-realizable in $\mathbb{R}^d$, in $O(|V(G)|^{1+\alpha})$ time for any constant $\alpha > 0$.
Here the constants hidden in $O(\cdot)$ depend on $k$, $d$, and $\alpha$.
\end{lemma}
\begin{proof}
Let $G_0$ be the corresponding undirected graph of $G$, and $\alpha > 0$ be a constant.
It suffices to compute $E \subseteq E(G_0)$ such that $|E| \leq \varepsilon |E(G_0)|$ and each connected component of $G_0 - E$ contains at most $(\frac{1}{\varepsilon})^{O(1)}$ vertices.
Indeed, the edges in $G$ corresponding to $E$ form a subset of $E(G)$ of size $O(\varepsilon |E(G)|)$ which satisfies the desired property.
This is fine since our algorithm works for an arbitrary $\varepsilon > 0$.

Pick a constant $\alpha' \in (0,\alpha)$ and let $\textsc{ApprxMinCut}(H)$ denote the algorithm in Theorem~\ref{thm-cut}, which returns an approximate minimum balanced cut of $H$ in $O(|V(H)|^{1+\alpha'})$ time.
Our algorithm, presented in Algorithm~\ref{alg-cutedge}, computes $E \subseteq E(G_0)$ by recursively applying $\textsc{ApprxMinCut}$ as a sub-routine.
We fix some threshold $\delta = (\frac{1}{\varepsilon})^c$ for a sufficiently large $c$ (depending on $k$, $d$, and $\alpha$).
If $|V(G_0)| \leq \delta$, then we simply return $E = \emptyset$.
Otherwise, we apply $\textsc{ApprxMinCut}(G_0)$ to obtain an approximate minimum balanced cut of $G_0$, and add the edges in the cut to $E$.
Then for every connected component $C$ of $G_0 - E$, we recursively apply our algorithm on $C$, which returns a subset of $E_C \subseteq E(C)$, and we add the edges in $E_C$ to $E$.
Finally, we return $E$ as the output of our algorithm.

\begin{algorithm}[htbp]
    \caption{\textsc{CutEdge}$(G_0)$}
    \begin{algorithmic}[1]
        \If{$|V(G_0)| \leq \delta$}{ \textbf{return} $\emptyset$}
            \begin{nolinenumbers}
            \end{nolinenumbers}
        \EndIf
        \State $E \leftarrow \textsc{ApprxMinCut}(G_0)$
        \State $\mathcal{C} \leftarrow$ set of connected components of $G_0 - E$
        \For{every $C \in \mathcal{C}$}
            \State $E \leftarrow E \cup \textsc{CutEdge}(C)$
            \begin{nolinenumbers}
            \end{nolinenumbers}
        \EndFor
        \State \textbf{return} $E$
    \end{algorithmic}
    \label{alg-cutedge}
\end{algorithm}

Clearly, every connected component of $G_0 - E$ contains at most $\delta = (\frac{1}{\varepsilon})^{O(1)}$ vertices.
So it suffices to show $|E| \leq \varepsilon |E(G_0)|$ and analyze the running time of the algorithm.
To bound $|E|$, we observe that when applying $\textsc{ApprxMinCut}$ on any subgraph $H$ of $G_0$, the size of the edge set obtained is of size $|V(H)|^{1-\frac{1}{d}} \log^{O(1/\alpha')} |V(H)|$.
Indeed, a minimum balanced cut of $H$ has size $O(|V(H)|^{1-\frac{1}{d}})$ by Lemma~\ref{lem-smallcut}, and $\textsc{ApprxMinCut}(H)$ returns a $\log^{O(1/\alpha')} |V(H)|$-approximate minimum balanced cut of $H$.
Since we defined $\delta = (\frac{1}{\varepsilon})^c$ for a sufficiently large $c$, we may assume that the size of $\textsc{ApprxMinCut}(H)$ is at most $|V(H)|^{1-\frac{1}{2d}}$ for any subgraph $H$ of $G_0$ such that $|V(H)| > \delta$.
Also, we may assume that $\varepsilon (2^{\frac{1}{3d}}-1) n^{1-\frac{1}{3d}} > n^{1-\frac{1}{2d}}$ for all $n > \delta$.
We shall prove that when applying our algorithm on a subgraph $H$ of $G_0$ with $|V(H)|=n$, the size of $E \subseteq E(H)$ returned is at most $\varepsilon (n - n^{1-\frac{1}{3d}})$.
We apply induction on $n$.
If $V(H) \leq \delta$, then our algorithm returns an empty set and thus the statement trivially holds.
Assume the statement holds for $|V(H)| \in [n-1]$, and we consider the case $|V(H)| = n$ (where $n > \delta$).
Our algorithm applies $\textsc{ApprxMinCut}(H)$ to obtain an approximate minimum balanced cut $E_0 \subseteq E(H)$ of $H$.
As aforementioned, $|E_0| \leq n^{1-\frac{1}{2d}}$.
Let $C_1,\dots,C_r$ be the connected components of $H - E_0$.
Set $n_i = |V(C_i)|$ for $i \in [r]$.
We have $n_i \leq n/2$ for all $i \in [r]$, by the definition of a balanced cut.
We recursively call our algorithm on each $C_i$ to obtain a subset $E_i \subseteq E(C_i)$.
By our induction hypothesis, $|E_i| \leq \varepsilon(n_i - n_i^{1-\frac{1}{3d}})$.
Finally, the algorithm returns $E = \bigcup_{i=0}^r E_i$.
Thus, we have $|E| = |E_0| + \sum_{i=1}^r |E_i|$.
As $\sum_{i=1}^r n_i = n$, $\sum_{i=1}^r |E_i| \leq \varepsilon n - \varepsilon \sum_{i=1}^r n_i^{1-\frac{1}{3d}}$.
Furthermore, since $n_i \leq n/2$ for all $i \in [r]$ and $\sum_{i=1}^r n_i = n$, it holds $\sum_{i=1}^r n_i^{1-\frac{1}{3d}} \geq 2(\frac{n}{2})^{1-\frac{1}{3d}} = 2^{\frac{1}{3d}} n^{1-\frac{1}{3d}}$.
Combining the bounds for $|E_0|$ and $\sum_{i=1}^r |E_i|$, we have
\begin{equation*}
    |E| \leq n^{1-\frac{1}{2d}} + \varepsilon n - \varepsilon \cdot 2^{\frac{1}{3d}} n^{1-\frac{1}{3d}}.
\end{equation*}
Recall our assumption $\varepsilon (2^{\frac{1}{3d}}-1) n^{1-\frac{1}{3d}} \geq n^{1-\frac{1}{2d}}$.
Together with the inequality above, this gives us $|E| \leq \varepsilon (n - n^{1-\frac{1}{3d}})$, so our induction works.
In particular, when we apply Algorithm~\ref{alg-cutedge} on $G_0$, the output $E \subseteq E(G_0)$ satisfies $|E| \leq \varepsilon (|V(G_0)| - |V(G_0)|^{1-\frac{1}{3d}}) \leq \varepsilon |V(G_0)|$.

To analyze the time complexity of Algorithm~\ref{alg-cutedge}, we consider the recursion tree $T$ of our algorithm when applying on $G_0$.
Each node $t \in T$ corresponds to a recursive call, and denote by $G_0(t)$ the corresponding graph handled in that call (which is a subgraph of $G_0$).
The time cost at a node $t \in T$ is $O(|E(G_0(t))|^{1+\alpha'})$, which is just $O(|V(G_0(t))|^{1+\alpha'})$ because $|E(G_0(t))| \leq k|V(G_0(t))|$.
Since $\textsc{ApprxMinCut}$ always returns a balanced cut, we know that $|V(G_0(t))| \leq |V(G_0(t'))|/2$ if $t'$ is the parent of $t$ in $T$.
This implies that the depth of $T$ is $O(\log |V(G_0)|)$.
Also, the sum of $|V(G_0(t))|$ for all nodes $t \in T$ at one level of $T$ is at most $|V(G_0)|$.
Thus $\sum_{t \in T} |V(G_0(t))| \leq |V(G_0)| \log |V(G_0)|$ and $\sum_{t \in T} |V(G_0(t))|^{1+\alpha'} \leq |V(G_0)|^{1+\alpha'} \log |V(G_0)|$.
The latter implies that the time complexity of Algorithm~\ref{alg-cutedge} is $O(|V(G_0)|^{1+\alpha'} \log |V(G_0)|)$, which is $O(|V(G_0)|^{1+\alpha})$ since $\alpha > \alpha'$.
We remark that a more careful analysis can be applied to show that the running time of our algorithm is actually $O(|V(G_0)|^{1+\alpha'})$ instead of $O(|V(G_0)|^{1+\alpha'} \log |V(G_0)|)$.
But for simplicity, we chose this looser analysis, which is already sufficient for our purpose.
\end{proof}

\subsection{Computing the approximate realization} \label{sec-component}

After computing the set $E \subseteq E(G)$ of edges using Lemma~\ref{lem-edgeremoval}, we consider the graph $G - E$.
Let $C_1,\dots,C_r$ be the weakly-connected components of $G - E$.
We have $|V(C_i)| = (\frac{1}{\varepsilon})^{O(1)}$ by Lemma~\ref{lem-edgeremoval}.
For each $i \in [r]$, we want to compute a ``$k$NN-realization'' of $C_i$ in $\mathbb{R}^d$.
Of course, here each $C_i$ is not necessarily $k$-regular, and thus a $k$NN-realization of $C_i$ is not defined.
But we can slightly generalize the definition of $k$NN-realization as follows.
For a directed graph $C$, we say a map $\phi: V(C) \rightarrow \mathbb{R}^d$ is a \textit{quasi-$k$NN-realization} of $C$ in $\mathbb{R}^d$ if for every edge $(u,v) \in E(C)$, $|\{v' \in V(C) \backslash \{u\}: \lVert \phi(u) - \phi(v') \rVert_2 \leq \lVert \phi(u) - \phi(v) \rVert_2\}| \leq k$.
By this definition, a $k$NN-realization is also a quasi-$k$NN-realization.
Also, it is clear that if $\phi:V(H) \rightarrow \mathbb{R}^d$ is a quasi-$k$NN-realization of a directed graph $H$ in $\mathbb{R}^d$, then for any subgraph $C$ of $H$, the map $\phi_{|V(C)}$ is a quasi-$k$NN-realization of $C$ in $\mathbb{R}^d$.
Therefore, if $G$ is $k$NN-realizable in $\mathbb{R}^d$, then each of $C_1,\dots,C_r$ admits a quasi-$k$NN-realization in $\mathbb{R}^d$.

Next, we discuss how to compute a quasi-$k$NN-realization of each $C_i$ in $\mathbb{R}^d$ (or conclude it does not exist).
To this end, we need the following lemma.

\begin{lemma}
Let $C$ be a directed graph where $|V(C)| \geq k+1$.
If $C$ admits a quasi-$k$NN-realization in $\mathbb{R}^d$, then there exists a $k$-regular supergraph $C'$ of $C$ with $V(C') = V(C)$ such that $C'$ is $k$NN-realizable in $\mathbb{R}^d$.
\end{lemma}
\begin{proof}
Assume $\phi: V(C) \rightarrow \mathbb{R}^d$ is a quasi-$k$NN-realization of $C$ in $\mathbb{R}^d$.
We can slightly perturb $\phi$ so that for any distinct $u,v,v' \in V(C)$, $\lVert \phi(u) - \phi(v') \rVert_2 \neq \lVert \phi(u) - \phi(v) \rVert_2\}$.
Now define $C'$ as a directed graph with $V(C') = V(C)$ and $E(C') = \{(u,v) \in V(C') \times V(C'): u \neq v \text{ and } \mathsf{rank}_\phi(u,v) \leq k\}$, where 
\begin{equation*}
    \mathsf{rank}_\phi(u,v) = |\{v' \in V(C') \backslash \{u\}: \lVert \phi(u) - \phi(v') \rVert_2 \leq \lVert \phi(u) - \phi(v) \rVert_2\}|.
\end{equation*}
The construction guarantees that $C'$ is $k$-regular and $\phi$ is a $k$NN-realization of $C'$ in $\mathbb{R}^d$.
Since $\phi$ is a quasi-$k$NN-realization of $C$, $E(C) \subseteq E(C')$ and thus $C'$ is a supergraph of $C$.
\end{proof}

If $|V(C_i)| \leq k$, then any map from $V(C_i)$ to $\mathbb{R}^d$ is a quasi-$k$NN-realization of $C_i$.
Otherwise, by the above lemma, to compute a quasi-$k$NN-realization of $C_i$ in $\mathbb{R}^d$, it suffices to consider every supergraph $C_i'$ of $C_i$ with $V(C_i') = V(C_i)$ and try to compute a $k$NN-realization of $C_i'$ (or conclude that $C_i'$ is not $k$NN-realizable).
Note that the number of such supergraphs is at most $\exp({(\frac{1}{\varepsilon})^{O(1)}})$ because $|V(C_i)| = (\frac{1}{\varepsilon})^{O(1)}$.
If we obtain a $k$NN-realization of some $C_i'$, then it is a quasi-$k$NN-realization of $C_i$.
To compute a $k$NN-realization $\phi:V(C_i') \rightarrow \mathbb{R}^d$ of $C_i'$, we formulate the problem as finding a solution to a system of $(\frac{1}{\varepsilon})^{O(1)}$ degree-2 polynomial inequalities on $(\frac{1}{\varepsilon})^{O(1)}$ variables.
For each $v \in V(C_i')$, we represent the coordinates of $\phi(v)$ by $d$ variables $x_1(v),\dots,x_d(v)$.
Then for all distinct $u,v,v' \in V(C_i')$ such that $(u,v) \in E(C_i')$ and $(u,v') \notin E(C_i')$, we introduce a degree-2 polynomial inequality $\sum_{j=1}^d (x_j(u)-x_j(v))^2 \leq \sum_{j=1}^d (x_j(u)-x_j(v'))^2$, which expresses $\lVert \phi(u) - \phi(v) \rVert_2 < \lVert \phi(u) - \phi(v') \rVert_2$.
Clearly, the solutions to this system of inequalities one-to-one correspond to the $k$NN-realizations of $C_i'$ in $\mathbb{R}^d$.
Renegar~\cite{renegar:1992} showed that a system of $p$ degree-2 polynomial inequalities on $q$ variables can be solved in $p^{O(q)}$ time.
Therefore, we can compute in $(\frac{1}{\varepsilon})^{(\frac{1}{\varepsilon})^{O(1)}}$ time a $k$NN-realization of $C_i'$ in $\mathbb{R}^d$ (or decide its non-existence).
This further implies that we can compute in $(\frac{1}{\varepsilon})^{(\frac{1}{\varepsilon})^{O(1)}}$ time a quasi-$k$NN-realization of $C_i$ in $\mathbb{R}^d$ (or decide its non-existence).
The total time cost for all $C_i$ is then $(\frac{1}{\varepsilon})^{(\frac{1}{\varepsilon})^{O(1)}} \cdot n$, since $r \leq n$.

If $C_i$ does not admit a quasi-$k$NN-realization in $\mathbb{R}^d$ for some $i \in [r]$, then we can directly conclude that $G$ is not $k$NN-realizable in $\mathbb{R}^d$.
Otherwise we construct a $(1-\varepsilon)$-approximate $k$NN-realization of $G$ in $\mathbb{R}^d$ as follows.
Let $\phi_i:V(C_i) \rightarrow \mathbb{R}^d$ be the quasi-$k$NN-realization of $C_i$ in $\mathbb{R}^d$ we compute, for $i \in [r]$.
Define a map $\phi:V(G) \rightarrow \mathbb{R}^d$ as follows.
Pick a vector $\vec{x} \in \mathbb{R}^d$ such that $\lVert \vec{x} \rVert_2 >> \max_{i \in [r]} \max_{u,v \in V(C_i)} \lVert \phi_i(u) - \phi_i(v) \rVert_2$.
Then for each $i \in [r]$ and each $v \in V(C_i)$, set $\phi(v) = i \vec{x} + \phi_i(v)$.
The choice of $\vec{x}$ guarantees that for a vertex $v \in V(C_i)$, the $\phi$-images of the vertices in $C_i$ are closer to $\phi(v)$ than the $\phi$-images of the vertices outside $C_i$.
Also, for any $u,v \in V(C_i)$, $\lVert \phi(u) - \phi(v) \rVert_2 = \lVert \phi_i(u) - \phi_i(v) \rVert_2$.
Therefore, for any $u,v \in V(C_i)$, we have
\begin{align*}
    & \ \{v' \in V(G) \backslash \{u\}: \lVert \phi(u) - \phi(v') \rVert_2 \leq \lVert \phi(u) - \phi(v) \rVert_2\} \\
    = & \ \{v' \in V(C_i) \backslash \{u\}: \lVert \phi(u) - \phi(v') \rVert_2 \leq \lVert \phi(u) - \phi(v) \rVert_2\} \\
    = & \ \{v' \in V(C_i) \backslash \{u\}: \lVert \phi_i(u) - \phi_i(v') \rVert_2 \leq \lVert \phi_i(u) - \phi_i(v) \rVert_2\}.
\end{align*}
Recall the function $\sigma_\phi:V(G) \times V(G) \rightarrow \{0,1\}$ in Definition~\ref{def-apprxreal}.
The above equality shows that for $u,v \in V(C_i)$, if $|\{v' \in V(C_i) \backslash \{u\}: \lVert \phi_i(u) - \phi_i(v') \rVert_2 \leq \lVert \phi_i(u) - \phi_i(v) \rVert_2\}| \leq k$, then $\sigma_\phi(u,v) = 1$.
It follows that $\sigma_\phi(u,v) = 1$ for any $(u,v) \in E(C_i)$ because $\phi_i$ is a quasi-$k$NN-realization of $C_i$,
so we have
\begin{equation*}
    \sum_{(u,v) \in E(G)} \sigma_\phi(u,v) \geq \sum_{i=1}^r |E(C_i)| = |E(G)| - |E| \geq (1-\varepsilon) \cdot |E(G)|.
\end{equation*}
Therefore, $\phi$ is a $(1-\varepsilon)$-approximate $k$NN-realization of $G$ in $\mathbb{R}^d$.
Combining the time costs for computing $E$ and the maps $\phi_1,\dots,\phi_r$, the overall time complexity of our algorithm is $O(n^{1+\alpha})$, where $O(\cdot)$ hides a constant depending on $k$, $d$, $\alpha$, and $\varepsilon$.
\begin{theorem}
Let $\alpha > 0$ be any fixed number.
Given a $k$-regular directed graph $G$ of $n$ vertices and a number $\varepsilon > 0$, one can compute in $f(k,d,\varepsilon) \cdot n^{1+\alpha}$ time a $(1-\varepsilon)$-approximate $k$NN-realization of $G$ in $\mathbb{R}^d$, or conclude that $G$ is not $k$NN-realizable in $\mathbb{R}^d$, where $f$ is some computable function depending on $\alpha$.
\end{theorem}

\newcommand{\coeff}[1]{C_{#1}}
\newcommand{\partcoeff}[2]{C_{#2}^{#1}}
\newcommand{\yleft}{y^-}
\newcommand{\yright}{y^+}

\section{\texorpdfstring{$k$}{k}NN-Realization in \texorpdfstring{$\mathbb{R}^1$}{R1}} \label{sec-1D}

To the best of our knowledge, the $k$NN-realization problem on the line does not seem to have been studied, and it is the focus of this section. A number of line embedding problems are $NP$-hard as mentioned earlier, including triplet constraints and betweenness problems~\cite{ChorSudan}. In the former, we are given a set of triplet constraints of the form $d(a,b) < d(a,c)$, while in the latter each constraint $(a,b,c)$ requires the ordering to satisfy $a < b < c$. Given the hardness result, the focus in these problems is on approximating the maximum number of satisfied constraints in the linear ordering~\cite{fan2020}. 

Unlike these intractable embedding problems on the line, we show that the partial order constraints implied by a $k$NN-realization problem have sufficiently  rich structure to admit a polynomial time solution. Most of the difficulty is in \emph{deciding} whether a $k$-regular directed graph $G$ is $k$NN-realizable on the line; if the answer is yes, then one can easily compute the embedding in polynomial time using linear programming.
(It is also worth pointing out that the decision problem for $k$NN-realization on the line is significantly more complicated than the special case of triplet constraints in \cite{fan2020} when \emph{all $\binom{n}{3}$ triples are specified}. Indeed, in that case, one can guess the leftmost point, and then all the remaining points are immediately determined by the triplet constraints. This is not the case in $k$NN-realization: fixing the leftmost point does not fix its neighbors' order.) 

The \emph{decision} version of the $k$NN-realization problem, of course, is easy \emph{if} we are given a permutation of $G$'s vertices $(v_1,\dots,v_n)$. In this case, we can easily compute a $k$NN-realization $\phi: V(G) \rightarrow \mathbb{R}^1$ satisfying $\phi(v_1) < \cdots < \phi(v_n)$, or decide that a feasible realization does not exist, by formulating the  problem as a linear program (LP) with $n$ variables $x_1,\dots,x_n$ and the following set of constraints
\begin{itemize}
    \item $x_i \leq x_j$ for all $i,j \in [n]$ with $i \leq j$,
    \item $\Delta_{i,j} < \Delta_{i,k}$ for all distinct $i,j,k \in [n]$ such that $(v_i,v_j) \in E(G)$ and $(v_i,v_k) \notin E(G)$, where $\Delta_{i,j} = x_{\max\{i,j\}} - x_{\min\{i,j\}}$ and $\Delta_{i,k} = x_{\max\{i,k\}} - x_{\min\{i,k\}}$.
\end{itemize}
Clearly, if $x_1,\dots,x_n$ satisfy these constraints, then setting $\phi(v_i) = x_i$ gives us the desired realization.
On the other hand, if $\phi$ is the realization we want, then the numbers $x_1,\dots,x_n$ where $x_i = \phi(v_i)$ must satisfy the constraints.
Therefore, we begin with this key problem: efficiently computing an ordering of the images of the vertices on $\mathbb{R}^1$.

\subsection{Finding the vertex ordering}

We say an ordering $(v_1,\dots,v_n)$ of $V(G)$ is a \textit{feasible vertex ordering} of $G$ if there exists a $k$NN-realization $\phi: V(G) \rightarrow \mathbb{R}^1$ of $G$ in $\mathbb{R}^1$ satisfying $\phi(v_1) < \cdots < \phi(v_n)$.
The goal of this section is to give an $O(kn)$-time algorithm for computing a feasible vertex ordering of $G$ (assuming it exists).
We may assume, without loss of generality, that $G$ is weakly-connected; otherwise we can consider each weakly-connected component of $G$ individually. 

Our first observation states two key properties of a feasible vertex ordering: (1) for each $v_i$ its $k$ out-neighbors form a contiguous block, and for different $v_i$'s their blocks have the same linear ordering as the vertex ordering, and (2) for each $v_i$ its in-neighbors also form a continuous block, which extends at most $k$ vertices to the left and at most $k$ to the right of $v_i$. (Recall that both $\mathsf{out}[v]$ and $\mathsf{in}[v]$ include $v$, for all vertices $v \in V(G)$.) See figure \ref{fig:obs1d}.

\begin{observation} \label{obs-vertex}
    A feasible vertex ordering $(v_1,\dots,v_n)$ of $G$ satisfies the following.
    \begin{enumerate}
        \item There exist $p_1,\dots,p_n \in [n-k]$ such that \textnormal{\bf (i)} $\mathsf{out}[v_i] = \{v_{p_i},\dots,v_{p_i+k}\}$ for all $i \in [n]$, \textnormal{\bf (ii)} $p_1 \leq \cdots \leq p_n$, and \textnormal{\bf (iii)} $p_i \leq i \leq p_i+k$ for all $i \in n$.
        \item There exist $q_1,\dots,q_n,q_1',\dots,q_n' \in [n]$ such that \textnormal{\bf (i)} $\mathsf{in}[v_i] = \{v_{q_i},\dots,v_{q_i'}\}$ for all $i \in [n]$, \textnormal{\bf (ii)} $q_1 \leq \cdots \leq q_n$, \textnormal{\bf (iii)} $q_1' \leq \cdots \leq q_n'$, and \textnormal{\bf (iv)} $q_i \leq q_i' \leq q_i+2k$ for all $i \in [n]$.
    \end{enumerate}
\end{observation}
\begin{proof}
Let $\phi:V(G) \rightarrow \mathbb{R}^1$ be a $k$NN-realization of $G$ in $\mathbb{R}^1$ satisfying $\phi(v_1) < \cdots < \phi(v_n)$.
Clearly, for each $i \in [n]$, the $k+1$ points in $\{\phi(v_1),\dots,\phi(v_n)\}$ closest to $\phi(v_i)$ are $\{\phi(v_{p_i}),\dots,\phi(v_{p_i+k})\}$ for some $p_i \in [n-k]$ such that $p_i \leq i \leq p_i+k$.
Furthermore, one can easily verify that $p_1 \leq \cdots \leq p_n$.
Assume $p_i > p_j$ for some $i,j \in [n]$ with $i<j$.
Then we have $\phi(v_{p_j}) < \phi(v_{p_i}) \leq \phi(v_i) < \phi(v_j) \leq \phi(v_{p_j+k}) < \phi(v_{p_i+k})$.
On the other hand, since $v_{p_i+k} \notin \{\phi(v_{p_j}),\dots,\phi(v_{p_j+k})\}$ and $v_{p_j} \notin \{\phi(v_{p_i}),\dots,\phi(v_{p_i+k})\}$,
\begin{equation*}
    \phi(v_j) - \phi(v_{p_j}) \leq \phi(v_{p_i+k}) - \phi(v_j) \leq \phi(v_{p_i+k}) - \phi(v_i) \leq \phi(v_i) - \phi(v_{p_j}),
\end{equation*}
which contradicts the fact $\phi(v_{p_j}) < \phi(v_{p_i}) \leq \phi(v_i) < \phi(v_j)$.
By the definition of $k$NN-realization, we see $\mathsf{out}[v_i] = \{v_{p_i},\dots,v_{p_i+k}\}$ for all $i \in [n]$.
This proves condition~1.

Condition~2 follows from condition~1.
Observe that for each $i \in [n]$, it holds that $\mathsf{in}[v_i] = \{v_j: p_j \leq i \leq p_j+k \} = \{v_j: i-k \leq p_j \leq i\}$.
Since $p_1 \leq \cdots \leq p_n$, this implies $\mathsf{in}[v_i] = \{v_{q_i},\dots,v_{q_i'}\}$ where $q_i = \min\{j: i-k \leq p_j \leq i\}$ and $q_i' = \max\{j: i-k \leq p_j \leq i\}$.
The monotoncities $q_1 \leq \cdots \leq q_n$ and $q_1' \leq \cdots \leq q_n'$ follow directly from the monotoncity $p_1 \leq \cdots \leq p_n$.
It suffices to show that $q_i \leq q_i' \leq q_i+2k$ for all $i \in [n]$.
The inequality $q_i \leq q_i'$ is trivial.
To see $q_i' \leq q_i+2k$, assume that $q_i' > q_i+2k$.
Then there exists $j,j' \in [n]$ with $j' > j+2k$ such that $i-k \leq p_j \leq i$ and $i-k \leq p_{j'} \leq i$.
By condition 1, we have $j \geq p_j$ and $p_{j'} \geq j' - k > j+k$.
Therefore, $p_{j'} > j+k \geq p_j+k \geq (i-k)+k = i$, which contradicts the fact $i-k \leq p_{j'} \leq i$.
This proves condition~2.
\end{proof}

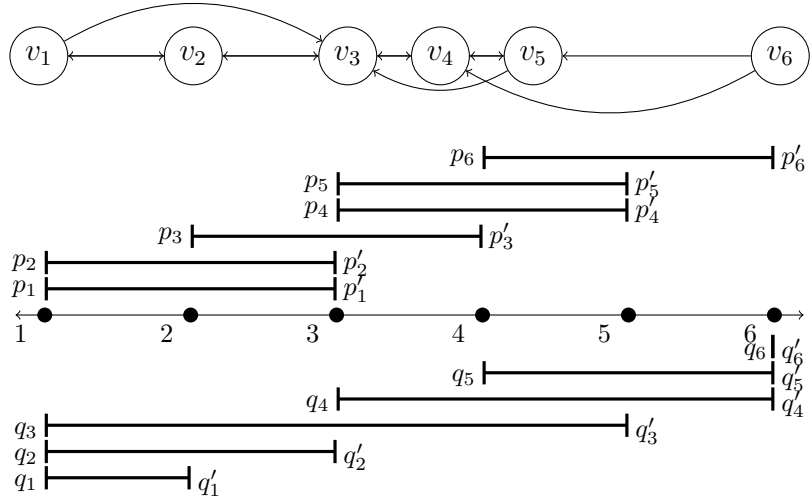
\begin{figure}
\begin{center}
\begin{scaletikzpicturetowidth}{0.75\textwidth}
\begin{tikzpicture}[ 
    main node/.style={circle,draw,font=\sffamily\Large\bfseries},
    scale=\tikzscale
]
\node[main node] (v1) at (0,0) {$v_1$};
\node[main node] (v2) at (2.5,0) {$v_2$};
\node[main node] (v3) at (5,0) {$v_3$};
\node[main node] (v4) at (6.5,0) {$v_4$};
\node[main node] (v5) at (8,0) {$v_5$};
\node[main node] (v6) at (12,0) {$v_6$};

\path[every node/.style={font=\sffamily\small}]
    (v1) edge[->] node {} (v2)
    (v1) edge[->, bend left] node {} (v3)
    (v2) edge[->] node {} (v1)
    (v2) edge[->] node {} (v3)
    (v3) edge[->] node {} (v2)
    (v3) edge[->] node {} (v4)
    (v4) edge[->] node {} (v3)
    (v4) edge[->] node {} (v5)
    (v5) edge[->, bend left] node {} (v3)
    (v5) edge[->] node {} (v4)
    (v6) edge[->, bend left] node {} (v4)
    (v6) edge[->] node {} (v5);
\end{tikzpicture}
\end{scaletikzpicturetowidth}

\begin{scaletikzpicturetowidth}{0.75\textwidth}
\begin{tikzpicture}[scale=\tikzscale, yscale=0.9,
    main node/.style={circle,draw,font=\sffamily\Large\bfseries}, 
    point/.style={circle, fill, inner sep=2pt},
    lbl/.style={circle, inner sep=0pt}
]

\draw[<->,step=0.5cm] (0,3) -- (13.5,3);

\node[point] (1) [label=185:{$1$}] at (0.5,3) {};
\node[point] (2) [label=185:{$2$}] at (3,3) {};
\node[point] (2) [label=185:{$3$}] at (5.5,3) {};
\node[point] (2) [label=185:{$4$}] at (8,3) {};
\node[point] (2) [label=185:{$5$}] at (10.5,3) {};
\node[point] (2) [label=185:{$6$}] at (13,3) {};

\draw[|-|,very thick] (0.5, 3.5) -- (5.5, 3.5);
\draw[|-|,very thick] (0.5, 4) -- (5.5, 4);
\draw[|-|,very thick] (3, 4.5) -- (8, 4.5);
\draw[|-|,very thick] (5.5, 5) -- (10.5, 5);
\draw[|-|,very thick] (5.5, 5.5) -- (10.5, 5.5);
\draw[|-|,very thick] (8, 6) -- (13, 6);

\draw[|-|,very thick] (0.5, 0-.1) -- (3, 0-.1);
\draw[|-|,very thick] (0.5, .5-.1) -- (5.5, .5-.1);
\draw[|-|,very thick] (0.5, 1-.1) -- (10.5, 1-.1);
\draw[|-|,very thick] (5.5, 1.5-.1) -- (13, 1.5-.1);
\draw[|-|,very thick] (8, 2-.1) -- (13, 2-.1);
\draw[|-|,very thick] (12.99999, 2.5-.1) -- (13, 2.5-.1);

\node[label] (p1) [label=left:{$p_1$}] at (0.7,3.5){};
\node[label] (p2) [label=left:{$p_2$}] at (0.7,4){};
\node[label] (p3) [label=left:{$p_3$}] at (3.2,4.5){};
\node[label] (p4) [label=left:{$p_4$}] at (5.7,5){};
\node[label] (p5) [label=left:{$p_5$}] at (5.7,5.5){};
\node[label] (p6) [label=left:{$p_6$}] at (8.2,6){};
\node[label] (pp1) [label=right:{$p_1'$}] at (5.3,3.5){};
\node[label] (pp2) [label=right:{$p_2'$}] at (5.3,4){};
\node[label] (pp3) [label=right:{$p_3'$}] at (7.8,4.5){};
\node[label] (pp4) [label=right:{$p_4'$}] at (10.3,5){};
\node[label] (pp5) [label=right:{$p_5'$}] at (10.3,5.5){};
\node[label] (pp6) [label=right:{$p_6'$}] at (12.8,6){};

\node[label] (q1) [label=left:{$q_1$}] at (0.7,0-.2){};
\node[label] (q2) [label=left:{$q_2$}] at (0.7,0.5-.2){};
\node[label] (q3) [label=left:{$q_3$}] at (0.7,1-.2){};
\node[label] (q4) [label=left:{$q_4$}] at (5.7,1.5-.2){};
\node[label] (q5) [label=left:{$q_5$}] at (8.2,2-.2){};
\node[label] (q6) [label=left:{$q_6$}] at (13.2,2.5-.2){};
\node[label] (qq1) [label=right:{$q_1'$}] at (2.8,0-.2){};
\node[label] (qq2) [label=right:{$q_2'$}] at (5.3,0.5-.2){};
\node[label] (qq3) [label=right:{$q_3'$}] at (10.3,1-.2){};
\node[label] (qq4) [label=right:{$q_4'$}] at (12.8,1.5-.2){};
\node[label] (qq5) [label=right:{$q_5'$}] at (12.8,2-.2){};
\node[label] (qq6) [label=right:{$q_6'$}] at (12.8,2.5-.2){};

\end{tikzpicture}
\end{scaletikzpicturetowidth}
\end{center}
\caption{A $2$-regular graph on $6$ vertices and its realization (top). 
The values of $p, p', q, q'$ for each vertex illustrate a natural ordering of the in and out-neighborhoods corresponding to the feasible vertex ordering $v_1, \cdots, v_6$ (bottom). 
}
\label{fig:obs1d}
\end{figure}

Consider the equivalence relation $\sim$ defined on $V(G)$ as $u \sim v$ iff $\mathsf{out}[u] = \mathsf{out}[v]$.
Let $\mathcal{C}_G$ be the set of equivalence classes of this relation, which is a partition of $V(G)$ into groups of vertices with the same set of out neighbors.
One can easily compute $\mathcal{C}_G$ in $O(k^2 n)$ time as follows. Given a vertex $v \in V(G)$, we check whether
$\mathsf{out}[u] = \mathsf{out}[v]$, for all $u \in \mathsf{out}[v]$, which takes $O(k^2)$ time since $|\mathsf{out}[v]| = k+1$ and each equality test takes $O(k)$ time. Thus, the time for computing all classes is $O(k^2 n)$.  Interestingly, we use Observation~\ref{obs-vertex} to compute $\mathcal{C}_G$ in $O(kn)$ time, if $G$ is $k$NN-realizable in $\mathbb{R}^1$, as shown in the following Lemma.

\begin{lemma} \label{lem-classify}
    Let $G$ be a $k$-regular directed graph of $n$ vertices.
    One can compute in $O(kn)$ time a partition $\mathcal{C}$ of $V(G)$ such that $\mathcal{C} = \mathcal{C}_G$ if $G$ is $k$NN-realizable in $\mathbb{R}^1$.
\end{lemma}

\begin{proof}
The algorithm for computing $\mathcal{C}$ is shown in Algorithm~\ref{alg-classify}.
Initially, we set $V = V(G)$ and $\mathcal{C} = \emptyset$.
In each iteration of the while-loop, we pick an arbitrary vertex $v \in V$ and try to partition the vertices in $I = \mathsf{in}[v] \cap V$ into classes.
\begin{algorithm}[t]
    \caption{\textsc{Classify}$(G)$}
    \begin{algorithmic}[1]
        \State $V \leftarrow V(G)$ and $\mathcal{C} \leftarrow \emptyset$
        \While{$V \neq \emptyset$}
            \State $v \leftarrow$ an arbitrary vertex in $V$
            \State $I \leftarrow \mathsf{in}[v] \cap V$
            \For{$i = 1,\dots,k+1$}
                \State $C,C' \leftarrow \emptyset$
                \For{every $u \in I$ such that $|\mathsf{out}[u] \cap \mathsf{out}[v]| = i$}
                    \If{$\mathsf{out}[u] = \mathsf{out}[x]$ for all $x \in C$}{ $C \leftarrow C \cup \{u\}$}
                    \Else{ $C' \leftarrow C' \cup \{u\}$}
                    \EndIf
                \EndFor
                \If{$C \neq \emptyset$}{ $\mathcal{C} \leftarrow \mathcal{C} \cup \{C\}$}
                \EndIf
                \If{$C' \neq \emptyset$}{ $\mathcal{C} \leftarrow \mathcal{C} \cup \{C'\}$}
                \EndIf
            \EndFor
            \State $V \leftarrow V \backslash I$
        \EndWhile
        \State \textbf{return} $\mathcal{C}$
    \end{algorithmic}
    \label{alg-classify}
\end{algorithm}
These classes are then added to $\mathcal{C}$ and the vertices in $I$ are removed from $V$.
We keep doing this until $V = \emptyset$, and finally return $\mathcal{C}$ as the output.
The set $I$ is partitioned as follows.
For each $i \in [k+1]$, we partition the vertices $u \in I$ with $|\mathsf{out}[u] \cap \mathsf{out}[v]| = i$ into two classes $C$ and $C'$.
The classification is done in a greedy manner: we consider these vertices one by one, and for each vertex $u$ considered, if $\mathsf{out}[u] = \mathsf{out}[x]$ for all $x \in C$, we add $u$ to $C$, otherwise we add $u$ to $C'$.
Then we add $C$ and $C'$ to $\mathcal{C}$ if they are nonempty.

To analyze the time complexity of Algorithm~\ref{alg-classify}, we show that each iteration of the while-loop takes $O(|\mathsf{in}[v]|+k|I|)$ time.
Computing $I$ can be done in $O(|\mathsf{in}[v]|)$ time.
After that, we compute $\mathsf{out}[u] \cap \mathsf{out}[v]$ for all $u \in I$, which takes $O(k|I|)$ time.
When considering a specific $u \in I$ with $|\mathsf{out}[u] \cap \mathsf{out}[v]| = i$, we need to test whether $\mathsf{out}[u] = \mathsf{out}[x]$ for all $x \in C$, in order to decide whether $u$ belongs to $C$ or $C'$.
Note that the construction of $C$ guarantees that all $x \in C$ have the same $\mathsf{out}[x]$.
Thus, we only need to pick an arbitrary $x \in C$ and test whether $\mathsf{out}[u] = \mathsf{out}[x]$, which takes $O(k)$ time.
Thus, each iteration of the while-loop takes $O(|\mathsf{in}[v]|+k|I|)$ time, and the overall running time of Algorithm~\ref{alg-classify} is $O(kn)$, because $\sum_{v \in V(G)} \mathsf{in}[v] = O(kn)$ and the sets $I$ in different iterations are disjoint.

Finally, we show the correctness of Algorithm~\ref{alg-classify}.
Observe that for each $v \in V(G)$, $\mathsf{in}[v]$ is the union of some classes in $\mathcal{C}_G$.
Indeed, for $C \in \mathcal{C}_G$, we have $C \subseteq \mathsf{in}[v]$ if $v \in \mathsf{out}[C]$ and $C \cap \mathsf{in}[v] = \emptyset$ if $v \not\in \mathsf{out}[C]$, where $\mathsf{out}[C] = \mathsf{out}[u]$ for arbitrary $u \in C$.
Thus, in each iteration of the while-loop, both $V$ and $I$ are unions of classes in $\mathcal{C}_G$.
It suffices to show that the set $I$ is partitioned correctly in each iteration.
Consider an iteration, and let $v$ be the vertex picked.
Let $(v_1,\dots,v_n)$ be a feasible vertex ordering of $G$, and suppose $v = v_j$.
Also, let $p_1,\dots,p_n \in [n-k]$ be the indices in condition~1 of Observation~\ref{obs-vertex}.
So we have $\mathsf{out}[v] = \{v_{p_j},\dots,v_{p_j+k}\}$
Clearly, for $u,u' \in I$, if $|\mathsf{out}[u] \cap \mathsf{out}[v]| \neq |\mathsf{out}[u'] \cap \mathsf{out}[v]|$, then $u$ and $u'$ do not belong to the same class in $\mathcal{C}_G$.
Algorithm~\ref{alg-classify} partitions all vertices $u \in I$ with $|\mathsf{out}[u] \cap \mathsf{out}[v]| = i$ into two classes $C$ and $C'$, where $\mathsf{out}[x] = \mathsf{out}[y]$ for all $x,y \in C$ and $\mathsf{out}[x] \neq \mathsf{out}[y]$ for all $x \in C$ and $y \in C'$.
By construction, we have $C \in \mathcal{C}_G$.
To see $C' \in \mathcal{C}_G$ as well, observe that if $|\mathsf{out}[u] \cap \mathsf{out}[v]| = i$ with $u = v_{j'}$ for some $j'$, then either $p_{j'} + k = p_j + i - 1$ and $\mathsf{out}[u] = \{v_{p_{j'}},\dots,v_{p_j+i-1}\}$, or $p_{j'} = p_j + k - i + 1$ and $\mathsf{out}[u] = \{v_{p_j+k-i+1},\dots,v_{p_{j'}+k}\}$, by Observation~\ref{obs-vertex}.
In other words, the vertices $u \in I$ with $|\mathsf{out}[u] \cap \mathsf{out}[v]| = i$ belong to (at most) two classes in $\mathcal{C}_G$.
As $C \in \mathcal{C}_G$, we must have $C' \in \mathcal{C}_G$ as well.
\end{proof}
Write $r = |\mathcal{C}_G|$.
Let $(v_1,\dots,v_n)$ be a feasible vertex ordering of $G$.
Observation~\ref{obs-vertex} implies that each class $C \in \mathcal{C}_G$ is a set of consecutive vertices in the sequence $(v_1,\dots,v_n)$, i.e., $C = \{v_\alpha,\dots,v_\beta\}$ for some $\alpha,\beta \in [n]$.
Define $\mathsf{out}[C] = \mathsf{out}[v]$ for an arbitrary vertex $v \in C$.
Therefore, there is a natural ordering $(C_1,\dots,C_r)$ of the classes in $\mathcal{C}_G$, in which the indices of the vertices in $C_i$ are smaller than the indices of the vertices in $C_j$ for all $i,j \in [r]$ with $i < j$.
We call $(C_1,\dots,C_r)$ a \textit{feasible class ordering} of $G$.

\begin{observation} \label{obs-class}
    If $(C_1,\dots,C_r)$ is the feasible class ordering of $G$ corresponding to a feasible vertex ordering $(v_1,\dots,v_n)$ of $G$, then there exist $c_1,\dots,c_r \in [n-k]$ satisfying
    \begin{itemize}
        \item $\mathsf{out}[C_i] = \{v_{c_i},\dots,v_{c_i+k}\}$ for all $i \in [r]$,
        \item $c_1 < \cdots < c_r$,
        \item if $G$ is weakly-connected, then $c_{i+1} \leq c_i+k$ for all $i \in [r-1]$.
    \end{itemize}
\end{observation}
\begin{proof}
Let $p_1,\dots,p_n \in [n-k]$ be the indices in condition~1 of Observation~\ref{obs-vertex}, for the feasible vertex ordering $(v_1,\dots,v_n)$.
For each $i \in [r]$, set $c_i = p_s$ where $s = 1+ \sum_{j=1}^{i-1} |C_j|$, and we then have $\mathsf{out}[C_i] = \mathsf{out}[v_s] = \{v_{c_i},\dots,v_{c_i+k}\}$.
Since $p_1 \leq \cdots \leq p_n$, we have $c_1 \leq \cdots \leq c_r$.
Furthermore, for all $i \in [r-1]$, we have $c_i \neq c_{i+1}$, simply because $\mathsf{out}[C_i] \neq \mathsf{out}[C_{i+1}]$.
Thus, $c_1 < \cdots < c_r$.
To see the last property, assume $c_{i+1} > c_i+k$ for some $i \in [r-1]$.
Then $(\bigcup_{j=1}^i \mathsf{out}[C_j]) \cap (\bigcup_{j=i+1}^r \mathsf{out}[C_j]) = \emptyset$.
Note that $\bigcup_{j=1}^i C_j \subseteq \bigcup_{j=1}^i \mathsf{out}[C_j]$ and $\bigcup_{j=i+1}^r C_j \subseteq \bigcup_{j=i+1}^r \mathsf{out}[C_j]$.
So there is no edge between $\bigcup_{j=1}^i C_j$ and $\bigcup_{j=i+1}^r C_j$, which implies that $G$ is not weakly connected.
\end{proof}

To compute a feasible vertex ordering of $G$, we first compute a feasible class ordering, using the previous observation.
\begin{lemma} \label{lem-classorder}
    Let $G$ be a weakly-connected $k$-regular directed graph of $n$ vertices.
    Given $\mathcal{C}_G$, one can compute in $O(kn)$ time an ordering of $\mathcal{C}_G$, which is a feasible class ordering of $G$ if $G$ is $k$NN-realizable in $\mathbb{R}^1$.
\end{lemma}

\begin{proof}
The algorithm for computing the ordering of $\mathcal{C}_G$ is shown in Algorithm~\ref{alg-classorder}.
It mainly works as follows.
We first arbitrarily pick a class $X \in \mathcal{C}_G$ and let $Y \in \mathcal{C}_G \backslash \{X\}$ that maximizes $|\mathsf{out}[X] \cap \mathsf{out}[Y]|$.
Now consider a (unknown) feasible class ordering $(C_1,\dots,C_r)$ of $G$.
Suppose $X = C_i$.
By Observation~\ref{obs-class}, either $Y = C_{i-1}$ or $Y = C_{i+1}$.
Without loss of generality, assume $Y = C_{i+1}$, for otherwise we can consider $(C_r,\dots,C_1)$, which is also a feasible class ordering of $G$.
The remaining task is to compute the sequences $(C_1,\dots,C_{i-1})$ and $(C_{i+1},\dots,C_r)$.
Let $\mathcal{C} = \mathcal{C}_G \backslash \{X\}$.
Define $\mathcal{L} \subseteq \mathcal{C}$ (resp., $\mathcal{R} \subseteq \mathcal{C}$) as the sub-collection consisting of all $C \in \mathcal{C}$ such that $\mathsf{out}[X] \cap \mathsf{out}[C] \neq \emptyset$ and $\mathsf{out}[X] \cap \mathsf{out}[C] \nsubseteq \mathsf{out}[X] \cap \mathsf{out}[Y]$ (resp., $\mathsf{out}[X] \cap \mathsf{out}[C] \subseteq \mathsf{out}[X] \cap \mathsf{out}[Y]$).
By Observation~\ref{obs-class}, $\mathcal{L}$ (resp., $\mathcal{R}$) consists of the classes in $C \in \{C_1,\dots,C_{i-1}\}$ (resp., $C \in \{C_{i+1},\dots,C_r\}$) satisfying that $\mathsf{out}[X] \cap \mathsf{out}[C] \neq \emptyset$.
To compute $(C_1,\dots,C_{i-1})$ and $(C_{i+1},\dots,C_r)$, we need a sub-routine $\textsc{Find}$ presented in Algorithm~\ref{alg-classorder}, which will be discussed below.
Calling $\textsc{Find}(\mathcal{C} \backslash \mathcal{L},X)$ is supposed to give us $(C_{i+1},\dots,C_r)$, while calling $\textsc{Find}(\mathcal{C} \backslash \mathcal{R},X)$ gives us $(C_{i-1},\dots,C_1)$.
Finally, we return the sequence $\mathsf{rev}(\textsc{Find}(\mathcal{C} \backslash \mathcal{R},X)) + (X) + \textsc{Find}(\mathcal{C} \backslash \mathcal{L},X)$, which is just $(C_1,\dots,C_r)$.
Here $\mathsf{rev}(\cdot)$ is the function that reverses a sequence.
Before discussing the sub-routine $\textsc{Find}$, we first show that computing $Y$, $\mathcal{L}$, $\mathcal{R}$ can be done in $O(k^2)$ time.
Indeed, Observation~\ref{obs-class} implies that there are only $O(k)$ classes $C \in \mathcal{C}$ satisfying $\mathsf{out}[X] \cap \mathsf{out}[C] \neq \emptyset$.
Also, $\mathsf{out}[X] \cap \mathsf{out}[C] \neq \emptyset$ iff $C \subseteq \bigcup_{v \in \mathsf{out}[X]} \mathsf{in}[v]$.
By condition~2 of Observation~\ref{obs-vertex}, $\sum_{v \in \mathsf{out}[X]} |\mathsf{in}[v]| = O(k^2)$.
Thus, the $O(k)$ classes $C \in \mathcal{C}$ satisfying $\mathsf{out}[X] \cap \mathsf{out}[C] \neq \emptyset$ can be computed in $O(k^2)$ time.
After this, we compute the intersection $\mathsf{out}[X] \cap \mathsf{out}[C]$ for every such class $C$, which takes $O(k^2)$ time in total.
Once these intersections are computed, we obtain $Y$ directly.
To further compute $\mathcal{L}$ and $\mathcal{R}$ takes another $O(k^2)$ time, since there are only $O(k)$ classes to be considered and one can decide whether each one belongs to $\mathcal{L}$ or $\mathcal{R}$ in $O(k)$ time.

\begin{algorithm}
    \caption{\textsc{ClassOrder}$(G,\mathcal{C}_G)$}
    \begin{algorithmic}[1]
        \State $\mathcal{C} \leftarrow \mathcal{C}_G$
        \State pick an arbitrary $X \in \mathcal{C}$ and $\mathcal{C} \leftarrow \mathcal{C} \backslash \{X\}$
        \State $Y \leftarrow \arg\max_{C \in \mathcal{C}} |\mathsf{out}[X] \cap \mathsf{out}[C]|$
        \State $\mathcal{L} \leftarrow \{C \in \mathcal{C}: \emptyset \neq \mathsf{out}[X] \cap \mathsf{out}[C] \nsubseteq \mathsf{out}[X] \cap \mathsf{out}[Y]\}$
        \State $\mathcal{R} \leftarrow \{C \in \mathcal{C}: \emptyset \neq \mathsf{out}[X] \cap \mathsf{out}[C] \subseteq \mathsf{out}[X] \cap \mathsf{out}[Y]\}$       
        \State \textbf{return} $\mathsf{rev}(\textsc{Find}(\mathcal{C} \backslash \mathcal{R},X))+(X)+\textsc{Find}(\mathcal{C} \backslash \mathcal{L},X)$
        \medskip
        \Function{\textsc{Find}}{$\mathcal{A},D$}
            \State $S \leftarrow \mathsf{empty\ list}$
            \While{$\mathsf{true}$}
                \State $\mathcal{A}' \leftarrow \{A \in \mathcal{A}: \mathsf{out}[D] \cap \mathsf{out}[A] \neq \emptyset\}$
                \State $t \leftarrow |\mathcal{A}'|$
                \If{$t = 0$}{ \textbf{break}}
                \EndIf                
                \State sort $\mathcal{A}'$ as $(A_1,\dots,A_t)$ such that $|\mathsf{out}[D] \cap \mathsf{out}[A_1]| \geq \cdots \geq |\mathsf{out}[D] \cap \mathsf{out}[A_t]|$
                \State $S \leftarrow S + (A_1,\dots,A_t)$
                \State $\mathcal{A} \leftarrow \mathcal{A} \backslash \mathcal{A}'$
                \State $D \leftarrow A_t$
            \EndWhile
            \State \textbf{return} $S$
        \EndFunction
    \end{algorithmic}
    \label{alg-classorder}
\end{algorithm}

Next, we discuss how the sub-routine $\textsc{Find}$ works.
We only consider the call $\textsc{Find}(\mathcal{C} \backslash \mathcal{L},X)$ and show it gives us $(C_{i+1},\dots,C_r)$; the other call $\textsc{Find}(\mathcal{C} \backslash \mathcal{R},X)$ is symmetric.
We set $\mathcal{A} = \mathcal{C} \backslash \mathcal{L}$ and $D = X$.
We iteratively compute a sequence $S$ of classes in $\mathcal{A}$ as follows.
Initially, $S$ is the empty list.
In every iteration, we first compute the sub-collection $\mathcal{A}' \subseteq \mathcal{A}$ consisting of all $A \in \mathcal{A}$ satisfying $\mathsf{out}[D] \cap \mathsf{out}[A] \neq \emptyset$.
Let $t = |\mathcal{A}'|$.
If $t = 0$ (i.e., $\mathcal{A}' = \emptyset$), we terminate the procedure and return the current $S$.
Otherwise, we sort the classes in $\mathcal{A}'$ as $(A_1,\dots,A_t)$ such that $|\mathsf{out}[D] \cap \mathsf{out}[A_1]| \geq \cdots \geq |\mathsf{out}[D] \cap \mathsf{out}[A_t]|$.
We append $(A_1,\dots,A_t)$ to the end of $S$.
Then we remove $\mathcal{A}'$ from $\mathcal{A}$, set $D$ to be $A_t$, and start the next iteration.
Note that this procedure will finally terminate.
Indeed, in all but the last iteration, we have $\mathcal{A}' \neq \emptyset$ and thus the size of $\mathcal{A}$ decreases at the end of the iteration.

To show $S = \{C_{i+1},\dots,C_r\}$ when the procedure terminates, suppose $\mathcal{L} = \{C_{p+1},\dots,C_{i-1}\}$ and thus $\mathcal{C} \backslash \mathcal{L} = \{C_1,\dots,C_p,C_{i+1},\dots,C_r\}$.
We prove the following invariant: at the beginning of each iteration, we have $D = C_j$, $S = (C_{i+1},\dots,C_j)$, and $\mathcal{A} = \{C_1,\dots,C_p,C_{j+1},\dots,C_r\}$, for some $j \geq i$.
Initially, the invariant holds for $j = i$.
Suppose at the beginning of some iteration, the invariant holds for $j$.
We consider how $D$, $S$, and $\mathcal{A}$ change in the iteration.
The sub-collection $\mathcal{A}'$ computed in this iteration contains all classes $A \in \{C_1,\dots,C_p,C_{j+1},\dots,C_r\}$ with $\mathsf{out}[C_j] \cap \mathsf{out}[A] \neq \emptyset$.
Since $\mathsf{out}[C_1],\dots,\mathsf{out}[C_p]$ are all disjoint from $\mathsf{out}[C_i]$, they are also disjoint from $\mathsf{out}[C_j]$ by Observation~\ref{obs-class}.
Thus, $\mathcal{A}' \subseteq \{C_{j+1},\dots,C_r\}$.
Furthermore, Observation~\ref{obs-class} implies that $\mathcal{A}' = \{C_{j+1},\dots,C_{j+t}\}$ for some $t \in [r-j]$, and also $\mathsf{out}[C_j] \cap \mathsf{out}[C_{j+1}] \supsetneq \mathsf{out}[C_j] \cap \mathsf{out}[C_{j+t}]$.
This implies that the sorted sequence $(A_1,\dots,A_t)$ in this iteration is equal to $(C_{j+1},\dots,C_{j+t})$.
At the beginning of the next iteration, we have $D = C_{j'}$, $S = (C_{i+1},\dots,C_{j'})$, and $\mathcal{A} = \{C_1,\dots,C_p,C_{j'+1},\dots,C_r\}$, where $j' = j+t$, and hence we still have the invariant.
This shows that the invariant holds throughout the procedure.
Now suppose we are at the beginning of the last iteration.
At this point, we have $D = C_j$, $S = (C_{i+1},\dots,C_j)$, and $\mathcal{A} = \{C_1,\dots,C_p,C_{j+1},\dots,C_r\}$ for some $j \geq i$, by the invariant.
We claim that $j = r$.
Assume $j < r$.
Then $C_{j+1} \in \mathcal{A}$.
As this is the last iteration, we must have $\mathcal{A}' = \emptyset$, i.e., $\mathsf{out}[C_j] \cap \mathsf{out}[A] = \emptyset$ for all $A \in \mathcal{A}$.
However, since $G$ is weakly-connected, Observation~\ref{obs-class} implies that $\mathsf{out}[C_j] \cap \mathsf{out}[C_{j+1}] \neq \emptyset$, which in turn implies $C_{j+1} \in \mathcal{A}'$ and contradicts the fact $\mathcal{A}' = \emptyset$.
Therefore, $j = r$ and $S = (C_{i+1},\dots,C_r)$ when the procedure terminates.
This proves the correctness of \textsc{Find}, as well as Algorithm~\ref{alg-classorder}.

Finally, we analyze the time complexity of \textsc{Find} (and the entire algorithm).
Suppose the construction of $S$ has $\ell$ iterations, and let $D_1,\dots,D_\ell$ (resp., $\mathcal{A}_1',\dots,\mathcal{A}_\ell'$) be the $D$ (resp., $\mathcal{A}'$) in the $\ell$ iterations.
In the $p$-th iteration, computing $\mathcal{A}_p'$ can be done in $O(\sum_{v \in \mathsf{out}[D_p]} |\mathsf{in}[v]|)$ time, because if $A \in \mathcal{A}_p'$, then $A \subseteq \mathsf{in}[v]$ for some $v \in \mathsf{out}[D_p]$.
Sorting the classes in $\mathcal{A}_p'$ takes $O(k+|\mathcal{A}_p'|)$, since here the sorting keys are integers in $[k+1]$.
So the time cost of the $p$-th iteration is $O(k+|\mathcal{A}_p'|+\sum_{v \in \mathsf{out}[D_p]} |\mathsf{in}[v]|)$.
The overall time cost is then $O(k \ell + \sum_{p=1}^\ell |\mathcal{A}_p'| + \sum_{p=1}^\ell \sum_{v \in \mathsf{out}[D_p]} |\mathsf{in}[v]|)$.
Note that $\ell \leq r \leq n$ and $\mathcal{A}_1',\dots,\mathcal{A}_\ell'$ are disjoint.
Thus, $k \ell + \sum_{p=1}^\ell |\mathcal{A}_p'| = O(kn)$.
To bound $\sum_{p=1}^\ell \sum_{v \in \mathsf{out}[D_p]} |\mathsf{in}[v]|$, the key is to observe that $\mathsf{out}[D_p] \cap \mathsf{out}[D_q] = \emptyset$ if $p-q>1$.
Indeed, $D_p \in \mathcal{A}_{p-1}'$ and hence $D_p \in \mathcal{A}$ at the end of the $q$-th iteration if $p-q>1$, implying that $\mathsf{out}[D_p] \cap \mathsf{out}[D_q] = \emptyset$.
It follows that each vertex $v \in V(G)$ appears in some $\mathsf{out}[D_p]$ at most twice.
So we have $\sum_{p=1}^\ell \sum_{v \in \mathsf{out}[D_p]} |\mathsf{in}[v]| \leq \sum_{v \in V(G)} 2|\mathsf{in}[v]| = O(kn)$.
As a result, \textsc{Find} can be implemented in $O(kn)$ time.
Recall that the time cost of Algorithm~\ref{alg-classorder} excluding the calls of \textsc{Find} is $O(k^2)$, as argued before.
We can then conclude that the time complexity of Algorithm~\ref{alg-classorder} is $O(kn)$.
\end{proof}

Next, we show how to compute a corresponding feasible vertex ordering of $G$ given a feasible class ordering $(C_1,\dots,C_r)$.
By definition, we know that in the feasible vertex ordering, the vertices in $C_i$ must appear before the vertices in $C_j$ for all $i,j \in [r]$ with $i < j$.
So it suffices to figure out the ``local'' ordering of the vertices in each class $C_i$.
We do this by considering the in-neighbors of each vertex.
As observed before, for each $v \in V(G)$, $\mathsf{in}[v]$ is the union of several classes in $\mathcal{C}_G$, and in addition, $\mathsf{in}[v] = \bigcup_{i=p}^q C_i$ for some $p,q \in [r]$, by condition~2 of Observation~\ref{obs-vertex}.
It turns out that we can sort the vertices in each class according to the values of $p$ and $q$.
Formally, we prove the following lemma.

\begin{lemma} \label{lem-vertexorder}
    Let $G$ be a weakly-connected $k$-regular directed graph of $n$ vertices, and $r = |\mathcal{C}_G|$.
    Given an ordering $(C_1,\dots,C_r)$ of $\mathcal{C}_G$, one can compute in $O(kn)$ time an ordering $(v_1,\dots,v_n)$ of $V(G)$ such that if $(C_1,\dots,C_r)$ is a feasible class ordering of $G$, then $(v_1,\dots,v_n)$ is a feasible vertex ordering of $G$.
\end{lemma}
\begin{proof}
Our algorithm for computing a feasible vertex ordering of $G$ is presented in Algorithm~\ref{alg-vertexorder}.
First, for each $v \in V(G)$, we compute a pair $R(v) = (p,q)$ where $p \in [r]$ (resp., $q \in [r]$) is the minimum (resp., maximum) index such that $C_p \subseteq \mathsf{in}[v]$ (resp., $C_q \subseteq \mathsf{in}[v]$).
We define a partial order $\preceq$ on index-pairs by setting $(p,q) \preceq (p',q')$ if $p+q \leq p'+q'$.
Then we order the vertices in each class $C_i$ as $(u_1,\dots,u_{|C_i|})$ such that $R(u_1) \preceq \cdots \preceq R(u_{|C_i|})$.
If there exist $u,v \in C_i$ such that $R(u)$ and $R(v)$ are incomparable under the order $\preceq$, then we just arbitrarily order the vertices in $C_i$.
We will see later that in this case, $(C_1,\dots,C_r)$ is not a feasible class ordering of $G$.
By concatenating the orderings of $C_1,\dots,C_r$, we obtain an ordering $(v_1,\dots,v_n)$ of $V(G)$.

We first show that Algorithm~\ref{alg-vertexorder} can be implemented in $O(kn)$ time.
Computing $R(v)$ for all $v \in V(G)$ can be done in $O(kn)$ time, because $|\mathsf{in}[v]| = O(k)$ by (ii) of Observation~\ref{obs-vertex}.
To see the time cost for ordering the vertices in each $C_i$, observe that $|C_i| \leq k$.
Indeed, $v \in \mathsf{out}[v] = \mathsf{out}[C_i]$ for all $v \in C_i$, which implies $C_i \subseteq \mathsf{out}[C_i]$ and thus $|C_i| \leq |\mathsf{out}[C_i]| = k$.
If we sort the vertices in $C_i$ directly, it takes $O(k \log k)$ time.
One can improve the time cost to $O(k)$ by observing that for any $v \in C_i$, the pair $R(v) = (p,q)$ satisfies $p \geq i-k$ and $q \leq i+k$, by condition~2 of Observation~\ref{obs-vertex}, which implies $2i-2k \leq p+q \leq 2i+2k$.
In other words, in the sorting task, the keys are all in the range $\{2i-2k,\dots,2i+2k\}$, whose size is $O(k)$.
Thus, the task can be done in $O(k)$ time using bucket sorting.
It follows that Algorithm~\ref{alg-vertexorder} can be implemented in $O(kn)$ time.

To see the correctness of our algorithm, suppose $(C_1,\dots,C_r)$ is a feasible class ordering of $G$, and let $(v_1^*,\dots,v_n^*)$ be a corresponding feasible vertex ordering of $G$.
We do not necessarily have $(v_1,\dots,v_n) = (v_1^*,\dots,v_n^*)$.
However, as we will see, it holds that $(R(v_1),\dots,R(v_n)) = (R(v_1^*),\dots,R(v_n^*))$, which turns out to be sufficient.
By (ii) of Observation~\ref{obs-vertex}, for each $v \in V(G)$ with $R(v) = (p,q)$, we have $\mathsf{in}[v] = \bigcup_{i=p}^q C_i$, and furthermore, $R(v_1^*) \preceq \cdots \preceq R(v_n^*)$.
Now consider a class $C_i$.
Note that $C_i = \{v_\alpha,\dots,v_\beta\} = \{v_\alpha^*,\dots,v_\beta^*\}$, where $\alpha = 1+\sum_{j=1}^{i-1} |C_j|$ and $\beta = \sum_{j=1}^i |C_j|$.
We have $R(v_\alpha^*) \preceq \cdots \preceq R(v_\beta^*)$, and our algorithm guarantees that $R(v_\alpha) \preceq \cdots \preceq R(v_\beta)$.
Therefore, $(R(v_\alpha),\dots,R(v_\beta)) = (R(v_\alpha^*),\dots,R(v_\beta^*))$.
It then follows that $(R(v_1),\dots,R(v_n)) = (R(v_1^*),\dots,R(v_n^*))$.

To see $(v_1,\dots,v_n)$ is a feasible vertex ordering of $G$, we further observe that the function $\pi:V(G) \rightarrow V(G)$ defined as $\pi(v_i) = v_i^*$ for $i \in [n]$ is an automorphism of $G$.
Consider indices $i,j \in [n]$.
We have $\mathsf{out}[v_i] = \mathsf{out}[v_i^*]$, since $v_i$ and $v_i^*$ belong to the same class in $\mathcal{C}_G$.
Thus, $(v_i,v_j) \in E(G)$ iff $(v_i^*,v_j) \in E(G)$.
On the other hand, because $R(v_j) = R(v_j^*)$, we have $\mathsf{in}[v_j] = \mathsf{in}[v_j^*]$ and hence $(v_i^*,v_j) \in E(G)$ iff $(v_i^*,v_j^*) \in E(G)$.
As such, $(v_i,v_j) \in E(G)$ iff $(v_i^*,v_j^*) \in E(G)$, which implies that $\pi$ is an automorphism of $G$.
Now consider a $k$NN-realization $\phi:V(G) \rightarrow \mathbb{R}^1$ of $G$ in $\mathbb{R}^1$ satisfying $\phi(v_1^*) < \cdots < \phi(v_n^*)$.
Set $\phi' = \phi \circ \pi$, which is a map from $V(G)$ to $\mathbb{R}^1$.
As $\pi$ is an automorphism of $G$, $\phi'$ is also a $k$NN-realization of $G$.
Furthermore, $\phi'(v_i) = \phi(\pi(v_i)) = \phi(v_i^*)$ for all $i \in [n]$, which implies $\phi'(v_1) < \cdots < \phi'(v_n)$.
The existence of $\phi'$ shows that $(v_1,\dots,v_n)$ is a feasible vertex ordering of $G$.
\end{proof}

\begin{algorithm}[htbp]
    \caption{\textsc{VertexOrder}$(G,(C_1,\dots,C_r))$}
    \begin{algorithmic}[1]
        \For{$v \in V(G)$}
            \State $p \leftarrow \min\{i \in [r]: C_i \subseteq \mathsf{in}[v]\}$
            \State $q \leftarrow \max\{i \in [r]: C_i \subseteq \mathsf{in}[v]\}$
            \State $R(v) \leftarrow (p,q)$
            \begin{nolinenumbers}
            \end{nolinenumbers}
        \EndFor
        \For{$i = 1,\dots,r$}
            \State $L_i \leftarrow$ an ordering $(u_1,\dots,u_{|C_i|})$ of $C_i$ satisfying $R(u_1) \preceq \cdots \preceq R(u_{|C_i|})$
            \begin{nolinenumbers}
            \end{nolinenumbers}
        \EndFor
        \State \textbf{return} $L_1+\cdots+L_r$
    \end{algorithmic}
    \label{alg-vertexorder}
\end{algorithm}

\begin{theorem} \label{thm-ordering}
Given a $k$-regular directed graph $G$ of $n$ vertices, one can compute in $O(kn)$ time an ordering $(v_1,\dots,v_n)$ of $V(G)$, which is a feasible vertex ordering of $G$ if $G$ is $k$NN-realizable in $\mathbb{R}^1$.
\end{theorem}
\begin{proof}
For the case where $G$ is weakly-connected, the theorem follows directly from Lemmas~\ref{lem-classify}, \ref{lem-classorder} and \ref{lem-vertexorder}.
Otherwise, we compute the orderings for the weakly-connected components of $G$ individually and concatenate them; this gives us the desired ordering of $V(G)$.
\end{proof}

\subsection{Deciding the realizability and computing the realization}

To decide the $k$NN-realizability of $G$, we first run the algorithm of Theorem~\ref{thm-ordering}, which returns the vertex ordering $(v_1,\dots,v_n)$ of $V(G)$, and then run the linear program given at the beginning of Section~\ref{sec-1D} for this ordering.
Either the LP is feasible, in which case the solution gives a $k$NN-realization of $G$ in $\mathbb{R}^1$, or LP is infeasible, in which case we can conclude that $(v_1,\dots,v_n)$ is not a feasible vertex ordering and thus $G$ is not $k$NN-realizable in $\mathbb{R}^1$. Thus, the $k$NN-realization problem in $\mathbb{R}^1$ can be solved in polynomial time.

Interestingly, we can show that if we are only interested in the \emph{decision} problem (and not actual embedding), then solving the LP is not necessary.
Specifically, we can decide whether $(v_1,\dots,v_n)$ is a feasible vertex by simply checking condition 1 of Observation~\ref{obs-vertex}.
If the ordering satisfies that condition, then the LP is \emph{always feasible}.

\begin{lemma} \label{lem-LP}
Let $G$ be a $k$-regular directed graph of $n$ vertices.
An ordering $(v_1,\dots,v_n)$ of $V(G)$ is a feasible vertex ordering of $G$ iff it satisfies condition 1 of Observation~\ref{obs-vertex}.
In particular, one can decide whether a given ordering of $V(G)$ is a feasible vertex ordering of $G$ or not in $O(kn)$ time.
\end{lemma}
\begin{proof}
The ``only if'' direction follows from Observation~\ref{obs-vertex}, so we only need to show the ``if'' direction.
Suppose $(v_1,\dots,v_n)$ satisfies condition 1 of Observation~\ref{obs-vertex}.
So there exist $p_1,\dots,p_n \in [n-k]$ such that \textnormal{\bf (i)} $p_1 \leq \cdots \leq p_n$, \textnormal{\bf (ii)} $p_i \leq i \leq p_i+k$ for all $i \in n$, and \textnormal{\bf (iii)} $\mathsf{out}[v_i] = \{v_{p_i},\dots,v_{p_i+k}\}$ for all $i \in [n]$.
To show $(v_1,\dots,v_n)$ is a feasible vertex ordering of $G$, it is equivalent to showing that the corresponding LP has a feasible solution.
It turns out that using the LP at the beginning of Section~\ref{sec-1D} is not very convenient.
So we shall formulate an alternative LP as follows.
For each $i \in [n-1]$, we introduce a variable $\Delta_i$, indicating the distance between $\phi(v_{i-1})$ and $\phi(v_i)$, where $\phi:V(G) \rightarrow \mathbb{R}^1$ is the $k$NN-realization we look for.
Then the distance between $\phi(v_\alpha)$ and $\phi(v_\beta)$ for $\alpha,\beta \in [n]$ with $\alpha \leq \beta$ can be represented as $\sum_{i=\alpha}^{\beta-1} \Delta_i$.
There are two types of constraints.
First, we need $\Delta_i > 0$ for all $i \in [n-1]$.
Second, for each $i \in [n]$, in order to guarantee that the $k$ nearest neighbors of $\phi(v_i)$ are $\phi(v_{p_i}),\dots,\phi(v_{p_i+k})$, we need the constraints $\sum_{j=p_i}^{i-1} \Delta_j < \sum_{j=i}^{p_i+k} \Delta_j$ and $\sum_{j=p_i-1}^{i-1} \Delta_j > \sum_{j=i}^{p_i+k-1} \Delta_j$.
These constraints guarantee that $\phi(v_{p_i})$ is closer to $\phi(v_i)$ than $\phi(v_{p_i+k+1})$ and $\phi(v_{p_i+k})$ is closer to $\phi(v_i)$ than $\phi(v_{p_i-1})$, which implies that the $k$ nearest neighbors of $\phi(v_i)$ are $\phi(v_{p_i}),\dots,\phi(v_{p_i+k})$, in the case $\phi(v_1) < \cdots < \phi(v_n)$.
Clearly, if this LP has a feasible solution $\Delta_1^*,\dots,\Delta_{n-1}^*$, then the map $\phi:V(G) \rightarrow \mathbb{R}^1$ defined as $\phi(v_i) = \sum_{j=1}^{i-1} \Delta_j^*$ for $i \in [n]$ is a $k$NN-realization of $G$ in $\mathbb{R}^1$.

We prove that the LP \textit{always} has a feasible solution, which implies that $(v_1,\dots,v_n)$ is a feasible vertex ordering of $G$.
We shall apply the well-known Farkas' lemma (specifically the variant (ii) of Proposition 6.4.3 in \cite{farkas-book}).
To this end, we need to slightly modify the constraints so that the LP can be expressed in the form $A \mathbf{x} \leq \mathbf{b}$ subject to $\mathbf{x} \geq 0$ for some matrix $A$ and vector $\mathbf{b}$.
First, we change the first type of constraints from $\Delta_i > 0$ to $\Delta_i \geq 0$.
Then we change the second type of constraints for each $v_i$ to $\sum_{j=p_i}^{i-1} \Delta_j - \sum_{j=i}^{p_i+k} \Delta_j \leq -1$ and $\sum_{j=i}^{p_i+k-1} \Delta_j - \sum_{j=p_i-1}^{i-1} \Delta_j \leq -1$.
Note that if the new LP has a feasible solution, then increasing each $\Delta_i$ by $\frac{1}{n}$ gives us a feasible solution of the original LP.
So it suffices to show that the new LP always has a feasible solution.
For convenience, we call $\sum_{j=p_i}^{i-1} \Delta_j - \sum_{j=i}^{p_i+k} \Delta_j \leq -1$ (resp., $\sum_{j=i}^{p_i+k-1} \Delta_j - \sum_{j=p_i-1}^{i-1} \Delta_j \leq -1$) the \textit{left} (resp., \textit{right}) \textit{constraint} for $v_i$.
If we encode the variables in a vector $\mathbf{x} = (\Delta_1,\dots,\Delta_{n-1})^T$, then each of the left/right constraints can be expressed as $\mathbf{a}^T \mathbf{x} \leq -1$ for some vector $\mathbf{a} \in \{-1,0,1\}^{n-1}$.
Therefore, the entire LP can be expressed as $A \mathbf{x} \leq \mathbf{b}$ subject to $\mathbf{x} \geq \mathbf{0}$, where $\mathbf{b} = (-1,\dots,-1)^T \in \mathbb{R}^{2n}$ and $A$ is a $2n \times (n-1)$ matrix whose row vectors are the $\mathbf{a}$-vectors for the left/right constraints.
Farkas' lemma states that there exists $\mathbf{x} \geq \mathbf{0}$ such that $A \mathbf{x} \leq \mathbf{b}$ iff there does \textit{not} exist a $2n$-dimensional vector $\mathbf{y} \geq \mathbf{0}$ such that $A^T \mathbf{y} \geq \mathbf{0}$ and $\mathbf{b}^T \mathbf{y} < 0$.
Note that if $\mathbf{y} \geq \mathbf{0}$, then $\mathbf{b}^T \mathbf{y} < 0$ is equivalent to saying that at least one entry of $\mathbf{y}$ is strictly positive, because $\mathbf{b} = (-1,\dots,-1)^T$.
Thus, it suffices to show the non-existence of such a vector $\mathbf{y}$.
Assume there exists $\mathbf{y} \geq \mathbf{0}$ such that $\mathbf{y} \neq \mathbf{0}$ and $A^T \mathbf{y} \geq \mathbf{0}$.
In the product $A^T \mathbf{y}$, each entry of $\mathbf{y}$ corresponds to a column vector of $A^T$, i.e., a row vector of $A$, which in turn corresponds to the left or right constraint of $v_i$ for some $i \in [n]$.
Let $y_i^-$ (resp., $y_i^+$) denote the entry of $y$ corresponds to the left (resp., right) constraint of $v_i$, for $i \in [n]$.
Define $t \in [n]$ as the smallest index such that $y_t^- > 0$ or $y_t^+ > 0$, which exists because $\mathbf{y} \geq \mathbf{0}$ and $\mathbf{y} \neq \mathbf{0}$.
For $i \in [n]$ and $j \in [n-1]$, we denote by $a_{i,j}^-$ (resp., $a_{i,j}^+$) the entry of $A$ in the row corresponding to $y_i^-$ (resp., $y_i^+$) and the column corresponding to the variable $\Delta_j$.
According to the constraints of the LP, we have
\begin{equation*}
    a_{i,j}^- = \left\{
    \begin{array}{ll}
        1 & \text{if } j \in \{p_i,\dots,i-1\}, \\
        -1 & \text{if } j \in \{i,\dots,p_i+k\}, \\
        0 & \text{otherwise},
    \end{array}
    \right.
    \text{ and \ }
    a_{i,j}^+ = \left\{
    \begin{array}{ll}
        1 & \text{if } j \in \{i,\dots,p_i+k-1\}, \\
        -1 & \text{if } j \in \{p_i-1,\dots,i-1\}, \\
        0 & \text{otherwise}.
    \end{array}
    \right.
\end{equation*}
We define $S_{i,j} = \sum_{x=1}^i (a_{x,j}^- y_x^- + a_{x,j}^+ y_x^+)$, for $i \in [n]$ and $j \in [n-1]$.
Clearly, for every $j \in [n-1]$, $S_{n,j}$ is the $j$-th entry of $A^T \mathbf{y}$, and hence $S_{n,j} \geq 0$.
On the other hand, we shall show below that $S_{n,j} < 0$ for some $j \in [n-1]$, which results in a contradiction.

Observe that for any $i \in [n]$ and $j \in [p_i-1]$, we have $S_{i,j} = S_{i+1,j} = \dots = S_{n,j}$.
Indeed, if $j < p_i$, then $j < p_{i'}$ for all $i' > i$ by the fact $p_1 \leq \cdots \leq p_n$, which implies that $a_{i',j}^- = a_{i',j}^+ = 0$ and thus $S_{i',j} = S_{i'-1,j}$.
Therefore, if there exists $i \in [n]$ and $j \in [p_i-1]$ such that $S_{i,j} < 0$, then $S_{n,j} < 0$ and we are done.
So assume $S_{i,j} \geq 0$ for all $i \in [n]$ and $j \in [p_i-1]$.
Under this assumption, we prove by induction that the following statement holds for all $i \in \{t,\dots,n\}$:
\begin{itemize}
    \item There exists $j \in [i]$ such that $S_{i,j} < 0$ and $S_{i,j'} < - S_{i,j}$ for any $j' \in \{i+1,\dots,n-1\}$.
\end{itemize}
The base case is $i = t$.
By the choice of $t$, we have $S_{t,j} = a_{t,j}^- y_t^- + a_{t,j}^+ y_t^+$ for all $j \in [n-1]$.
However, it must be that $y_t^+ = 0$, and thus $y_t^- > 0$.
Indeed, if $y_t^+ > 0$, consider for $j'' = p_t-1$ that $a_{t,j''}^- = 0$, $a_{t,j''}^+ = -1$, and thus $S_{t,j''} = -y_t^+ < 0$, which contradicts our assumption that $S_{i,j} \geq 0$ for $j \in [p_i - 1]$.
Now set $j = i$.
Then $a_{t,j}^- = -1$ and $a_{t,j}^+ = 1$.
Since $y_t^- > 0$ and $y_t^+ = 0$, we have $S_{t,j} = -y_t^- < 0$.
For any $j' \in \{i+1,\dots,n-1\}$, $S_{t,j'} = y_t^+-y_t^- = -y_t^-$ if $j' \leq p_t+k-1$, $S_{t,j'} = -y_t^-$ if $j' = p_t+k$, and $S_{t,j'} = 0$ if $j' > p_t+k$.
In all these cases, $S_{t,j'} \leq 0 < y_t^- = -S_{t,j}$.
So the statement holds for $i = t$.
Assume the statement holds for $i-1$.
Then there exists $j_0 \in [i-1]$ such that $S_{i-1,j_0} < 0$ and $S_{i-1,j'} < - S_{i-1,j_0}$ for any $j' \in \{i,\dots,n-1\}$.
We must have $j_0 \geq p_i$, because we have assumed that $S_{i,x} \geq 0$ for all $x \in [p_i-1]$.
As $p_i \leq j_0 \leq i-1$, we have $S_{i,j_0} = S_{i-1,j_0} + y_i^- - y_i^+$, or equivalently,
\begin{equation} \label{eq-Sij0}
    -S_{i,j_0} = -S_{i-1,j_0} - y_i^- + y_i^+.
\end{equation}
This equality will be used later.
To show the statement also holds for $i$, we first observe that if $S_{i,j} < 0$, then the condition $S_{i,j'} < -S_{i,j}$ for all $j' \in \{i+1,\dots,n-1\}$ holds directly.
Assume $S_{i,j} < 0$ and consider an index $j' \in \{i+1,\dots,n-1\}$.
The value of $S_{i,j'}$ can have the following three possibilities.
\begin{enumerate}
    \item $S_{i,j'} = S_{i-1,j'} + y_i^+ - y_i^-$ if $j' \leq p_i+k-1$.
    By our induction hypothesis, $S_{i-1,j'} < -S_{i-1,j}$.
    Thus, $S_{i,j'} < -S_{i-1,j} + y_i^+ - y_i^-$.
    By Equation~\ref{eq-Sij0}, this implies $S_{i,j'} < -S_{i,j}$.
    \item $S_{i,j'} = S_{i-1,j'}-y_i^-$ if $j' = p_i+k$.
    In this case, $S_{i,j'} \leq S_{i-1,j'} + y_i^+ - y_i^-$.
    So the same argument as in the first case shows $S_{i,j'} < -S_{i,j}$.
    \item $S_{i,j'} = S_{i-1,j'}$ if $j' > p_i+k$.
    Since $p_1 \leq \dots \leq p_i$, we have $j' > p_{i'} + k$ for all $i' \in [i]$.
    Thus, $S_{i,j'} = S_{i-1,j'} = \cdots = S_{0,j'} = 0$.
    As $S_{i,j} < 0$, we have $S_{i,j'} < -S_{i,j}$.
\end{enumerate}
So we always have $S_{i,j'} < -S_{i,j}$.
As such, it suffices to show the existence of $j \in [i]$ satisfying $S_{i,j} < 0$.
If $S_{i,j_0} < 0$, we can simply set $j = j_0$.
If $S_{i,j_0} \geq 0$, we set $j = i$.
We have $S_{i,i} = S_{i-1,i} - y_i^- + y_i^+$.
By our induction hypothesis, $S_{i-1,i} < -S_{i-1,j_0}$.
Thus, $S_{i,i} < -S_{i-1,j_0} - y_i^- + y_i^+$, which further implies $S_{i,i} < - S_{i,j_0} \leq 0$ by Equation~\ref{eq-Sij0}, i.e., $S_{i,j} < 0$.
As a result, the statement holds for $i$.
Applying the statement with $i = n$, we see the existence of $j \in [n]$ with $S_{n,j} < 0$, resulting in a contradiction.
Therefore, $(v_1,\dots,v_n)$ is a feasible vertex ordering of $G$ iff it satisfies condition~1 of Observation~\ref{obs-vertex}.

Finally, given an ordering $(v_1,\dots,v_n)$ of $V(G)$, one can easily test in $O(kn)$ time whether it satisfies condition~1 of Observation~\ref{obs-vertex}, which indicates whether it is a feasible vertex ordering of $G$.
This completes the proof of the lemma.
\end{proof}
\noindent
We can now state the main result of this section.

\begin{theorem} \label{thm:final}
Given a $k$-regular directed graph $G$ of $n$ vertices, one can decide in $O(kn)$ time whether $G$ is $k$NN-realizable in $\mathbb{R}^1$, and if so, a $k$NN-realization of $G$ in $\mathbb{R}^1$ can be computed in $O(n^{2.5} \mathsf{poly}(\log n))$ time.
\end{theorem}

\begin{proof}
We first use Theorem~\ref{thm-ordering} to compute an ordering $(v_1,\dots,v_n)$ of $V(G)$ in $O(kn)$ time. By Lemma~\ref{lem-LP}, we decide in $O(kn)$ time whether $(v_1,\dots,v_n)$ is a feasible vertex ordering of $G$. In this way, we know whether $G$ is $k$NN-realizable in $\mathbb{R}^1$ or not in $O(kn)$ time. If $G$ is realizable, we solve the LP in the proof of Lemma~\ref{lem-LP} to obtain a $k$NN-realization of $G$ in $\mathbb{R}^1$. The LP has $O(n)$ variables and $O(n)$ constraints, so it can be solved in $O(n^{2.5} \mathsf{poly} (\log n))$ time using, for instance, the algorithm of Lee and Sidford~\cite{lee2015efficient}.
\end{proof}

\section{Concluding remarks and extensions}

We considered the problem of realizing a directed graph $G$ as an Euclidean $k$NN graph. Our key results are: (1) for any fixed $d$, we can efficiently embed at least a $1 - \varepsilon$ fraction of $G$'s edges in $\mathbb{R}^d$ or conclude that $G$ is not realizable, and (2) a linear time algorithm to decide if $G$ is realizable in $\mathbb{R}^1$.
Our theorems extend to the case where the neighbors of each vertex in $G$ are given as a \emph{ranked} list, meaning that the embedding must satisfy $||\phi(v) - \phi(u_i)|| < || \phi(v) - \phi(u_{i+1})||$, for $i=1, \ldots, k-1$, where $u_i$ is the $i$th nearest neighbor of $v$ (except in $\mathbb{R}^1$ where we need to solve the LP to decide if it is feasible).
Our approximation scheme also applies to other proximity graphs that meet the following conditions: (1) the graph can be partitioned into constant-sized components using a sublinear size separator, and (2) each component's edges can be embedded independently. For example, we can approximately embed Delaunay triangulations in the plane with maximum degree $k = O(n^\frac{1}{3})$. 

\bibliographystyle{plainurl}
\bibliography{main}

\end{document}